\documentclass[11pt]{amsart}
\usepackage{fullpage}
\usepackage[foot]{amsaddr}
\usepackage{microtype}
\usepackage[OT1]{fontenc}
\usepackage{eulervm}
\usepackage[tt=false]{libertine} 
\usepackage{bbold}
\usepackage{amsmath}
\usepackage{amssymb}
\usepackage{amsthm}
\usepackage{thmtools} 
\usepackage[linesnumbered,boxed,ruled,vlined]{algorithm2e}
\usepackage{algpseudocode}
\usepackage{enumitem}
\usepackage{multirow}
\usepackage{bm}
\usepackage{xifthen}
\usepackage{xspace}
\usepackage{tikz}
\usetikzlibrary{arrows.meta,positioning,calc,fit,backgrounds}

\usepackage[margin=1cm]{caption} 
\usepackage{subfig}

\usepackage[thinlines]{easytable}

\usepackage[bookmarks=true,hypertexnames=false,pagebackref]{hyperref}
\hypersetup{colorlinks=true, citecolor=blue, linkcolor=red, urlcolor=blue}

\usepackage{pgfplots}
\pgfplotsset{compat=1.16}
\usepackage{tikz}
\usetikzlibrary{arrows,arrows.meta,backgrounds,calc,fit,decorations.pathreplacing,decorations.markings,shapes.geometric}

\tikzstyle{internal} = [draw, fill, shape=circle]
\tikzstyle{external} = [shape=circle]
\tikzstyle{square}   = [draw, fill, rectangle]
\tikzstyle{triangle} = [draw, fill, regular polygon, regular polygon sides=3, inner sep=3pt]
\tikzstyle{pentagon} = [draw, fill, regular polygon, regular polygon sides=5, inner sep=2pt, minimum size=14pt]
\tikzset{every fit/.append style=text badly centered}

\usetikzlibrary{positioning,chains,fit,shapes,calc}
\usetikzlibrary{trees}
\usetikzlibrary{decorations.pathreplacing}
\usetikzlibrary{decorations.pathmorphing}
\usetikzlibrary{decorations.markings}
\tikzset{>=latex} 

\usepackage{ifthen}

\usepackage{cleveref}

\usepackage[textsize=tiny]{todonotes}

\usepackage[normalem]{ulem}

\usepackage{mleftright}

\usepackage{cool}
\Style{DSymb={\mathrm d},DShorten=true,IntegrateDifferentialDSymb=\mathrm{d}}

\renewcommand{\Pr}{\mathop{\mathrm{Pr}}\nolimits}

\def\*#1{\mathbf{#1}}
\def\+#1{\mathcal{#1}}
\def\-#1{\mathrm{#1}}
\def\=#1{\mathbb{#1}}
\def\^#1{\mathbb{#1}}

\newcommand{\poly}[1]{\ensuremath{\mathop{\mathrm{poly}}\inp{#1}}}

\newcommand{\norm}[1]{\ensuremath{\left\lVert #1\right\rVert}}

\newcommand{\eps}{\varepsilon}

\newcommand{\dist}{\operatorname{dist}}
\newcommand{\diag}{\operatorname{diag}}

\newtheorem{theorem}{Theorem}

\newtheorem{lemma}[theorem]{Lemma}

\newtheorem{proposition}[theorem]{Proposition}

\theoremstyle{definition}

\newtheorem{definition}[theorem]{Definition}

\theoremstyle{remark}

\crefname{theorem}{Theorem}{Theorems}
\crefname{observation}{Observation}{Observations}
\crefname{claim}{Claim}{Claims}
\crefname{condition}{Condition}{Conditions}
\crefname{algorithm}{Algorithm}{Algorithms}
\crefname{property}{Property}{Properties}
\crefname{example}{Example}{Examples}
\crefname{fact}{Fact}{Facts}
\crefname{lemma}{Lemma}{Lemmas}
\crefname{corollary}{Corollary}{Corollaries}
\crefname{definition}{Definition}{Definitions}
\crefname{remark}{Remark}{Remarks}
\crefname{proposition}{Proposition}{Propositions}
\crefname{equation}{equation}{equations}
\crefname{enumi}{Case}{Case}
\creflabelformat{enumi}{(#2#1#3)}

\definecolor{HGcolor}{RGB}{255,50,50}

\makeatletter
\def\prob#1#2#3{\goodbreak\begin{list}{}{\labelwidth\z@ \itemindent-\leftmargin
      \itemsep\z@  \topsep6\p@\@plus6\p@
      \let\makelabel\descriptionlabel}
  \item[\textbf{Name}]#1
  \item[\textbf{Instance}]#2
  \item[\textbf{Output}]#3
  \end{list}}
\makeatother

\makeatletter
\providecommand\@dotsep{5}
\def\listtodoname{Todo list}
\def\listoftodos{\@starttoc{tdo}\listtodoname}
\makeatother

\newcommand{\TV}{d_{\mathrm{TV}}}
\newcommand{\R}{\mathbb{R}}
\newcommand{\F}{\mathbb{F}}

\newcommand{\ket}[1]{\lvert #1\rangle}
\newcommand{\bra}[1]{\langle #1\rvert}
\newcommand{\one}{\mathbf{1}}
\newcommand{\wt}{\operatorname{wt}}
\newcommand{\supp}{\operatorname{supp}}
\newcommand{\Cay}{\operatorname{Cay}}
\newcommand{\Prb}{\mathop{\mathrm{Pr}}}

\usepackage{nicefrac,comment}

\usepackage{silence}
\newboolean{doubleblind}
\setboolean{doubleblind}{false}

\begin{document}

\title{Classical and quantum spectral density estimation under local graph access}

\author{Rong-Hua Li}
\author{Meihao Liao}
\author{Yichun Yang}
\address[Rong-Hua Li]{School of Computer Science, Beijing Institute of Technology, Beijing, China}
\address[Meihao Liao]{School of Computer Science, Beijing Institute of Technology, Beijing, China}
\address[Yichun Yang]{School of Computer Science, Beijing Institute of Technology, Beijing, China}

\maketitle

\begin{abstract}
We study spectral density estimation for the normalized adjacency matrix of
an unweighted graph under local access model. Previously, Cohen-Steiner et al. [KDD 2018] proposed an algorithm for $\varepsilon$-approximate spectral density estimation in the Wasserstein-1 distance, using $2^{O(1/\varepsilon)}$ local queries to the graph.  In this paper, we prove that every constant-success estimator with
Wasserstein--$1$ error at most $\eps$ requires
$2^{\Omega(1/\eps)}$ queries, showing that the Cohen-Steiner algorithm is optimal up to constant in the exponent. This resolves the open problem left by previous researches Jin et al. [COLT 2023] and Peng et al. [COLT 2026].

We then turn to quantum local access model. We give
an $\widetilde O(\eps^{-3})$-query algorithm estimating the spectral density
with Wasserstein-1 error at most $\eps$. Finally, we prove a
$\widetilde\Omega(\eps^{-4/3})$ quantum lower bound when the graph is
sufficiently large.
As a result, quantum local access model changes the dependence on $\eps$ from exponential to polynomial.
\end{abstract}

\section{Introduction}

The empirical spectral distribution of a large graph encodes global
information about expansion, mixing, geometry, and the prevalence of closed
walks.  Computing it by diagonalizing the normalized adjacency matrix is
prohibitively expensive when the graph is too large to read.  This motivates
\emph{spectral density estimation}: output a succinct probability measure
that approximates the distribution on the eigenvalues while using
only local graph queries.  Spectral densities have found applications in various fields such as characterizing the structure of complex networks \cite{FarkasEtAl2001,EikmeierGleich2017,VillegasEtAl2023}, numerical linear algebra \cite{LinSaadYang2016,UbaruChenSaad2017} and (graph) machine learning
\cite{DefferrardEtAl2016,KipfWelling2017,GhorbaniEtAl2019,MahoneyMartin2019}.

Let $G=(V,E)$ be a finite simple unweighted graph, let $n=|V|$ and $m=|E|$, $A_G$ be the adjacency matrix and $D_G$ be the diagonal degree matrix, we write the normalized adjacency matrix as
\[
  N_G=D_G^{-1/2}A_GD_G^{-1/2}.
\]
Given the eigenvalues $\lambda_1 \ge \lambda_2\ge \ldots\ge \lambda_n$ of $N_G$, the spectral density is defined as
\[
  \rho_G=\frac1n\sum_{j=1}^n\delta_{\lambda_j}.
\]
We denote $\delta_{\lambda}$ to be the Dirac delta function centered at $\lambda$. We measure the computation error by the Wasserstein--$1$ distance $W_1$. That is, we want to output a probability density function $\tilde{\rho}$ such that $W_1(\tilde{\rho}, \rho_G) \leq \eps$. Equivalently \cite{CohenSteinerEtAl2018}, the goal is to compute (succinct representations of) estimates $\tilde{\lambda}_1 \ge \tilde{\lambda}_2\ge \ldots\ge \tilde{\lambda}_n$ such that 
\[\frac{1}{n}\sum_{i=1}^n |\lambda_i - \tilde{\lambda}_i| \leq \eps.\]

We note that the $W_1$ metric is weak enough to permit a succinct $O(1/\eps)$-point output \cite{CohenSteinerEtAl2018}, but strong enough to
control the average error of each eigenvalue estimation.

\subsection{Classical local access.}
Under classical local graph access model \cite{bender2002testing}, the algorithm is allowed to the following operations in $O(1)$ time: (i) sample a uniform vertex, (ii) query the degree $d_v$ of a vertex $v$, (iii) query the $i$-th neighbor of vertex $v$ with $i\in [d_v]$. In this direction,
Cohen-Steiner, Kong, Sohler, and Valiant~\cite{CohenSteinerEtAl2018}
gave an algorithm with
$2^{O(1/\eps)}$ query complexity via sampling random walks from uniformly random nodes. Subsequent lines of researches developed sublinear and
near-linear tradeoffs in this direction
\cite{BravermanKrishnanMusco2022,JinEtAl2023,JinEtAl2024,bhattacharjee2025improved,musco2025sharper,PengEtAl2026}. Specifically, \cite{BravermanKrishnanMusco2022} gave a $O(n\cdot \eps^{-7})$ sublinear time algorithm that $\eps$-approximates the spectral density. Their algorithm was later improved to $O(n\cdot \eps^{-2})$ query complexity and $O(n\cdot \eps^{-3})$ time by \cite{JinEtAl2024}. Comparison with Cohen-Steiner et al. \cite{CohenSteinerEtAl2018}, these algorithms denpends only polynomial on $\eps^{-1}$, but at the cost of linearly depends on $n$. Besides, there are also algorithms that use $O(\eps^{-1})$ queries to the matrix-vector multiplication to approximate the spectral density \cite{BravermanKrishnanMusco2022,bhattacharjee2025improved,musco2025sharper}. However, the query complexity of matrix-vector multiplication is $O(m)$ in the local access model.

From the lower bound perspective, the key question is to ask whether the exponential dependence on $\eps^{-1}$ is necessary. In this direction, Jin, Musco, Sidford, and Singh~\cite{JinEtAl2023} showed that no algorithm can compute an $\eps$-approximation of the spectral density under $W_1$ with constant success probability even when given $2^{\Omega(1/\eps)}$ random walks of length $2^{\Omega(1/\eps)}$ started from uniformly random nodes. However, their lower bound is only constraint to random walk sampling from uniformly random nodes, and is constraint to weighted graphs. Later, Peng, Wang, Yang, and Yang~\cite{PengEtAl2026} extend to a $2^{\Omega(1/\eps^{1/6})}$ lower bound for unweighted graphs when sampling random walks from uniformly random nodes. It was left open whether one can extend their lower bounds to the general local access model \cite{bender2002testing}.

Our first result resolves this question. So the algorithm proposed by Cohen-Steiner et al. \cite{CohenSteinerEtAl2018} is already optimal up to constant in the exponent under local access model when $\eps$ is small constant. Table~\ref{tab:classical-related} summarizes the main classical upper and lower bounds, including our result.

\begin{theorem}
\label{thm:intro-classical}
There are universal constants $c,\eps_0>0$ such that, for every
$0<\eps<\eps_0$, any randomized algorithm under local access model 
that outputs an $\eps$-approximation of spectral density with probability at least
$2/3$ on simple unweighted graph requires at least $2^{c/\eps}$
queries in the worst case. 
\end{theorem}

\begin{table}[ht]
    \centering
    \caption{Classical algorithms and lower bounds for $\eps$-approximate spectral density estimation under $W_1$. The algorithms are assumed with $\eps=\tilde{\Omega}(1/\sqrt{n})$. We denote matvec to be the matrix--vector multiplication, $\sigma_l$ to be the $l$-th largest singular value of $N_G$. Non-adaptive RW refers to random walks started from uniformly random nodes.}
    \label{tab:classical-related}
    \renewcommand{\arraystretch}{1.3}
    \small
    \begin{tabular}{@{}l l p{2.2 cm} p{2.5cm} c p{3.8cm}@{}}
    \hline
     & Algorithm & Complexity & Query model & Error & Graph Type \\
    \hline
    Local
      & \cite{CohenSteinerEtAl2018}
      & $2^{O(1/\eps)}$
      & local access
      & $W_1\le \eps$
      & Unweighted \\
    \hline
    \multirow{2}{*}{Sublinear}
      & \cite{BravermanKrishnanMusco2022}
      & $O(n\cdot\eps^{-7})$
      & local access
      & $W_1\le \eps$
      & Unweighted \\
      & \cite{JinEtAl2024}
      & $O(n\cdot\eps^{-3})$
      & local access
      & $W_1\le \eps$
      & Weighted or Unweighted\\
    \hline
    \multirow{2}{*}{Linear}
      & \cite{BravermanKrishnanMusco2022,musco2025sharper}
      & $O(m\cdot \eps^{-1})$
      & matvec
      & $W_1\le \eps$
      & Weighted or Unweighted \\
      & \cite{bhattacharjee2025improved}
      & $\widetilde{O}(m\cdot (\eps^{-1}+l))$
      & matvec
      & $W_1\le \eps\sigma_{l+1}$
      & Weighted or Unweighted \\
    \hline
    \multirow{3}{*}{Lower bound}
      & \cite{JinEtAl2023}
      & $2^{\Omega(1/\eps)}$
      & non-adaptive RW
      & $W_1\le \eps$
      & Weighted \\
      & \cite{PengEtAl2026}
      & $2^{\Omega(1/\eps^{1/6})}$
      & non-adaptive RW
      & $W_1\le \eps$
      & Unweighted \\
      & Ours
      & $2^{\Omega(1/\eps)}$
      & local access
      & $W_1\le \eps$
      & Unweighted \\
    \hline
    \end{tabular}
    \end{table}

\subsection{Quantum local access.}
We next extend the classical local access model to the quantum setting.
In the quantum model, vertices are represented by $\lceil\log n\rceil$ bits. The following classical graph functions are exposed as reversible XOR oracles and we assume they can be implemented in $O(1)$ time:
\begin{align*}
 O_{\rm deg}\ket{v,z}&=\ket{v,z\oplus d_v},\\
 O_{\rm nbr}\ket{v,i,z}&=\ket{v,i,z\oplus {\rm nbr}(v,i)},\\
 O_{\rm pair}\ket{u,v,z}&=\ket{u,v,z\oplus A_{uv}}.
\end{align*}
Here ${\rm nbr}(v,i)$ is the neighbor occupying port $i\in[d_v]$ in an
arbitrary but fixed adjacency-list ordering.  This model is typical in previous researches ~\cite{apers2022quantum,JiangPeng2026}; see \Cref{sec:prelim} for further
details.  We obtain the following result, which shows the spectral density can be estimated with query complexity only $\poly{1/\eps}$.
\begin{theorem}
\label{thm:intro-quantum-upper}
In the quantum local access model, there is an algorithm with $\widetilde O\!\left(\eps^{-3}\right)$ quantum query complexity and additional $\widetilde O\!\left(\eps^{-2}\right)$ classical runtime that $\eps$-approximates the spectral density under $W_1$ with probability at least $2/3$, where the $\widetilde O(\cdot )$ only hides the $\poly{\log \frac{1}{\eps}}$ factor.
\end{theorem}

We complement the algorithm by a quantum query lower bound.

\begin{theorem}
\label{thm:intro-quantum-lower}
Every spectral
density estimator in the quantum local access model with $W_1$ error at most
$\eps$ and success probability at least $2/3$ uses
$\widetilde\Omega(\eps^{-4/3})$ queries in the worst case.
The $\widetilde\Omega(\cdot)$ hides only $\poly{\log \frac{1}{\eps}}$ factor.
\end{theorem}


\subsection{Model comparison.}
The quantum local access model contains the classical local access model:
the three operations in the classic access model \cite{bender2002testing} can be implemented from the quantum coherent XOR oracles
at $O(1)$ cost.  The quantum model is therefore at least as strong, and in
fact strictly more powerful. The coherent superposition queries enable quantum walk-based
phase estimation that produces samples from the approximation of spectral
density $\rho_G$.  Our quantum upper bound is built on this capability.  By
contrast, \Cref{thm:intro-classical} shows that in the classical local access
model no algorithm can output an $\eps$-approximation of $\rho_G$ with
constant success probability using fewer than $2^{\Omega(1/\eps)}$ queries.
In particular, the classical model does not admit a comparably efficient
procedure for sampling from $\rho_G$.  Spectral density estimation therefore
exhibits a quantum advantage.

\subsection{AI disclaimer.} The main idea of the construction of lower bounds are inspired by the interaction with GPT-5.6 Sol. All proofs are human written or have been heavily revised by the authors. 

\section{Preliminaries}
\label{sec:prelim}

\subsection{Some basic definitions}

For a finite simple unweighted graph $G=(V,E)$, let $A=A_G$ be its adjacency
matrix and $D=D_G=\diag(d_v:v\in V)$ its degree matrix.  We use the convention
that $D^{-1/2}_{vv}=0$ when $d_v=0$, and define the normalized matrix $N=N_G=D^{-1/2}AD^{-1/2}$.
This real symmetric matrix $N$ has eigenvalues $\lambda_1 \ge \lambda_2\ge \ldots\ge \lambda_n$ in $[-1,1]$.

\begin{definition}
    The empirical spectral distribution of $G$ is defined by
\[
                 \rho_G=\frac1n\sum_{j=1}^n\delta_{\lambda_j}.
\]
Where $\lambda_j$ is the $j$-th eigenvalue of $N_G$, and $\delta_{\lambda_j}$ is the Dirac delta function centered at $\lambda_j$.
\end{definition}

\begin{definition}
\label{def:wasserstein-1}
For probability measures (densities) $\mu,\nu$ on $\R$, the Wasserstein--$1$ distance is
\[
 W_1(\mu,\nu)=\inf_{\gamma\in\Gamma(\mu,\nu)}
                  \int |x-y|\,d\gamma(x,y),
\]
where $\Gamma(\mu,\nu)$ is the set of couplings.
\end{definition}

Given a probability measure $\mu$, the cumulative distribution function (CDF) is defined as $F_\mu(x)=\mu((-\infty,x])$. We use the Kantorovich--Rubinstein duality \cite{Villani2009} and one dimension property \cite{Panaretos_2019} for the $W_1$ distance. 

\begin{proposition}\label{prop:kr-duality}
Let $\operatorname{Lip}(f)$ be the Lipschitz constant of a Lipschitz function $f$. Then the following holds for Wasserstein--$1$ distance:
\[
        W_1(\mu,\nu)=\int \left|F_\mu(x)-F_\nu(x)\right|dx=\sup_{\operatorname{Lip}(f)\leq1}
                 \left|\int f\,d\mu-\int f\,d\nu\right|.\label{eq:kr-duality}
\]
\end{proposition}

Next we provide some basic definitions about vector space, which will be frequently used in section \ref{sec:classical}.
Throughout, we work with finite-dimensional vector spaces over the binary field $\F_2$.
A nonempty subset $U\subseteq V$ is a \emph{subspace} if it is closed under addition (equivalently, under $\F_2$-linear combinations).
For $S\subseteq V$, we write $\langle S\rangle$ for the subspace spanned by $S$, and $\langle v\rangle$ when $S=\{v\}$.
The \emph{dimension} $\dim V$ is the cardinality of any basis of $V$, and if $U\subseteq V$ is a subspace then $\dim U\leq\dim V$.
We also recall the following standard definitions in linear algebra.

\begin{definition}[Complementary subspace]
    \label{def:complementary-subspace}
    Let $V$ be a vector space over $\F_2$ and let $U\subseteq V$ be a subspace.
    A subspace $W\subseteq V$ is a \emph{complement} of $U$ if
    \[
     V=U+W
     \qquad\text{and}\qquad
     U\cap W=\{0\},
    \]
    equivalently we write $V=U\oplus W$.  In this case
    $\dim W=\dim V-\dim U$.
    \end{definition}

\begin{definition}[Orthogonal complement]
\label{def:orthogonal-complement}
Let $V$ be a finite-dimensional vector space over $\F_2$ equipped with the
standard bilinear form
\[
 x\cdot y=\sum_{i}x_iy_i\in\F_2.
\]
For a subspace $U\subseteq V$, the \emph{orthogonal complement} of $U$ is
\[
 U^\perp=\{v\in V:v\cdot u=0\text{ for all }u\in U\}.
\]
In particular, for a vector $v\in V$ we write
$v^\perp=\langle v\rangle^\perp$.
\end{definition}

\begin{definition}[Quotient space]
\label{def:quotient-space}
Let $V$ be a vector space over $\F_2$ and let $U\subseteq V$ be a subspace.
The \emph{quotient space} $V/U$ is the vector space of cosets
\[
 V/U=\{x+U:x\in V\},
\]
with addition
\[
 (x+U)+(y+U)=(x+y)+U.
\]
We write $[x]=x+U$ for the coset of $x$.  If $\dim V<\infty$, then
\[
 \dim(V/U)=\dim V-\dim U.
\]
\end{definition}

\begin{definition}[Dual space]
\label{def:dual-space}
Let $V$ be a vector space over $\F_2$.  The \emph{dual space} $V^*$ is the
vector space of linear maps $x^*: V\to\F_2$, with pointwise addition.  If
$\dim V<\infty$, then $\dim V^*=\dim V$.  Given a basis
$e_1,\ldots,e_n$ of $V$, the dual basis $e_1^*,\ldots,e_n^*$ of $V^*$ is
determined by $e_i^*(e_j)=\delta_{ij}$.
\end{definition}

\subsection{Quantum computation background.}
In quantum computing, the state of a system is represented by a unit vector
$\ket{v}$ in a Hilbert space $\mathcal{H}$, which is a complex inner product
vector space.  For a $d$-dimensional Hilbert space $\mathbb{C}^d$, the
standard computational basis consists of orthonormal vectors
$\{\ket{i}\}_{i=0}^{d-1}$, where
$\ket{i}=(0,\ldots,0,1,0,\ldots,0)^{\top}$ denotes the $i$th basis vector.
A single qubit is represented as $\alpha\ket{0}+\beta\ket{1}$ with complex
amplitudes $\alpha,\beta$ satisfying $|\alpha|^2+|\beta|^2=1$.  Multi-qubit
states are represented via tensor products: for
$\ket{a},\ket{b}\in\mathbb{C}^d$, their tensor product is their Kronecker
product
$\ket{a}\otimes\ket{b}\equiv\ket{a}\ket{b}
=(a_0b_0,a_0b_1,\ldots,a_1b_0,a_1b_1,\ldots,a_{d-1}b_{d-1})^{\top}
\in\mathbb{C}^{d^2}$.

The evolution of a closed quantum system is described by a unitary operator
$U$ satisfying $U^\dagger U=I$, acting as $\ket{\psi}\mapsto U\ket{\psi}$,
where $U^\dagger$ is the Hermitian conjugate of $U$, and $I$ is the identity
matrix.  To extract classical information from state $\ket{\psi}$,
measurement in the computational basis yields outcome $i$ with probability
$|\langle\psi|i\rangle|^2$, collapsing the state to $\ket{i}$, where
$\langle\psi|i\rangle$ is the inner product of $\ket{\psi}$ and $\ket{i}$.

In the quantum model, vertices are represented by $\lceil\log n\rceil$ bits
and the following classical functions are exposed as reversible XOR oracles:
\begin{align*}
 O_{\rm deg}\ket{v,z}&=\ket{v,z\oplus d_v},\\
 O_{\rm nbr}\ket{v,i,z}&=\ket{v,i,z\oplus {\rm nbr}(v,i)},\\
 O_{\rm pair}\ket{u,v,z}&=\ket{u,v,z\oplus A_{uv}}.
\end{align*}
Here ${\rm nbr}(v,i)$ is the neighbor occupying port $i\in[d_v]$ in an
arbitrary but fixed adjacency-list ordering.  These are the coherent local
queries used in previous works~\cite{apers2022quantum,JiangPeng2026}.  Controlled queries cost only a constant
factor: for an XOR oracle one may compute the answer into a clean register,
conditionally copy it, and uncompute it with a second query.

We count calls to the graph oracles and allow arbitrary quantum computation
and classical postprocessing between queries.  The final answer is a
classical, explicitly represented probability distribution.

\section{An exponential lower bound under classic model}
\label{sec:classical}

We now prove \Cref{thm:intro-classical}, the exponential lower bound under the classic model.  The hard
graphs are quotients of the hypercube by binary self-dual codes. At
accuracy $\eps$, both hard graphs consist of simple regular graphs of
common degree $d=\Theta(1/\eps)$ and common order $n=2^{d/2}$. \Cref{fig:cayley-quotient} illustrates the construction. Form the spectrum aspect, Type-II
codes force all spectral atoms onto one lattice, whereas Type-I codes place
exactly half of the atoms halfway between consecutive lattice points. \Cref{fig:type-spectra} illustrates the spectral densities for Type-I/Type-II graphs for $d=16$. We prove that the spectral densities of the two types of graphs are $\eps$-far from each other (\Cref{prop:classical-separation}). From the query lower bound aspect, we show that the
random code remains invisible until the exploration discovers a nontrivial
code relation. We utilize some probability bounds from finite geometry (\Cref{subsec:finite-geometry-incidence-bounds}) to prove that this is exponentially
unlikely under local graph queries exploration (\Cref{subsec:indistinguishability-via-coupling}). 

\subsection{A quadratic space from self-dual codes}

We first begin with our construction of the quadratic space from self-dual codes, which will be used for our hard case construction. The construction of self-dual codes is typical in finite geometry \cite{rains2002self}, and we include a detailed construction here for completeness. In this section, all vector spaces are over $\F_2$. We refer to Section \ref{sec:prelim} for some basic definitions.

We fix a sufficiently large integer $d$ such that
\[
                              d\equiv0\pmod 8.
\]
Let $h=\frac d2-1$. Let $\one\in\F_2^d$ be the all-one vector, and put
\[
 R=\langle\one\rangle,
 \qquad E=\one^\perp,
 \qquad W=E/R.
\]

In particular, the subspace $E$
defined above is called the even-weight subspace, and
$\dim E=d-1$.  Consequently $\dim W=d-2=2h$. We provide the following definition.

\begin{definition}
\label{def:hamming-weight}
For a vector $x=(x_1,\ldots,x_d)\in\F_2^d$, the \emph{Hamming weight} of $x$ is
\[
 \wt(x)=\bigl|\{i\in[d]:x_i=1\}\bigr|,
\]
Equivalently,
$\wt(x)=\sum_{i=1}^d x_i$ when the sum is taken in $\mathbb{Z}$.
\end{definition}

\begin{definition}
\label{def:even-weight}
The \emph{even-weight subspace} of $\F_2^d$ is
\[
 \{x\in\F_2^d:\wt(x)\equiv0\pmod 2\}.
\]
Equivalently, with respect to the standard bilinear form on $\F_2^d$, this
subspace coincides with $E=\one^\perp$.  
\end{definition}

For our analysis, we use the following definitions.

\begin{definition}
\label{def:quadratic-form}
Let $V$ be a finite-dimensional vector space over $\F_2$.  A function
$q:V\to\F_2$ is a \emph{quadratic form} if the associated map
\[
 b(x,y)=q(x+y)+q(x)+q(y),\qquad x,y\in V,
\]
is a bilinear form on $V$. The form $b$ is called the \emph{polar form} of
$q$. $(V,q)$ is called a \emph{quadratic space}. The \emph{radical} of $b$ is
\[
 \operatorname{rad}(b)=\{x\in V:b(x,y)=0\text{ for all }y\in V\}.
\]
We say that $q$ is \emph{nondegenerate} if $\operatorname{rad}(b)=\{0\}$.
A subspace $U\subseteq V$ is \emph{totally singular} (with respect to $q$) if
$q(x)\equiv0$ for all $x\in U$.
\end{definition}


\begin{lemma}
\label{lem:quadratic-quotient}
For $x\in E$, the formula
\[
                    q([x])=\frac{\wt(x)}2\pmod2
\]
defines a nondegenerate quadratic form on $W$.  Its polar form is
\[
 b([x],[y])=q([x+y])+q([x])+q([y])=x\cdot y.
\]
Moreover, $(W,q)$ has a totally singular subspace of dimension $h$.
\end{lemma}

\begin{proof}
Every $x\in E$ has even weight. That is, $\wt(x)\equiv0\pmod 2$.  Replacing $x$ by $x+\one$ changes the weight
to $d-\wt(x)$, and $d/2$ is even.  Hence
\[
 \frac{d-\wt(x)}2\equiv\frac{\wt(x)}2\pmod2,
\]
so $q$ is well defined on $E/R$.  The identity
\[
 \wt(x+y)=\wt(x)+\wt(y)
       -2|\supp(x)\cap\supp(y)|
\]
gives the following polar form:
\begin{align*}
 b([x],[y])&=\frac{1}{2}\left[\wt(x+y)+\wt(x)+\wt(y)\right]\pmod2\\
 &=\wt(x)+\wt(y)-|\supp(x)\cap\supp(y)|\pmod2\\
 &=|\supp(x)\cap\supp(y)|\pmod2\\
 &=x\cdot y
\end{align*}
The dot product restricted to
$E=\one^\perp$ has radical exactly $R$, so its quotient on $W$ is
nondegenerate.

Finally, we show that $(W,q)$ has a totally singular subspace of dimension $h$. We consider the following:
\[A=
 \begin{pmatrix}
 1&1&1&1&0&0&0&0\\
 1&1&0&0&1&1&0&0\\
 1&0&1&0&1&0&1&0\\
 1&1&1&1&1&1&1&1
 \end{pmatrix}
\]
Then $A$ generates a linear subspace of dimension $4$ on $\F_2^8$.  The direct sum of $d/8$
copies of $A$ generates a linear subspace of dimension $d/2$ on $\F_2^d$.  Its image in $W$ is a
totally singular subspace of dimension $d/2-1=h$.  
\end{proof}

\begin{definition}
\label{def:self-dual-code}
A linear subspace $C\subseteq\F_2^d$ is \emph{self-dual} if $C=C^\perp$. We call $C$ a \emph{self-dual code}. In particular,
every self-dual code in $\F_2^d$ has dimension $d/2$.
\end{definition}

\begin{definition}
\label{def:symplectic-space}
A \emph{symplectic space} over $\F_2$ is a pair $(V,b)$, where $V$ is a
finite-dimensional vector space of even dimension and
$b:V\times V\to\F_2$ is a nondegenerate alternating bilinear form
(i.e., $b(x,x)=0$ for all $x\in V$, and
$\operatorname{rad}(b)=\{0\}$).  In particular, if $(V,q)$ is a
nondegenerate quadratic space of even dimension, then $(V,b)$ is symplectic
for the polar form $b$ of $q$.
\end{definition}

\begin{definition}
\label{def:lagrangian}
Let $(V,b)$ be a $2h$-dimensional symplectic space over $\F_2$.  A subspace
$L\subseteq V$ is \emph{isotropic} if $b(x,y)=0$ for all $x,y\in L$.  An
isotropic subspace of dimension $h$ is called a \emph{Lagrangian}.
\end{definition}

\begin{definition}
\label{def:type-I-II}
Let $C\subseteq\F_2^d$ be a self-dual code.  We say that $C$ is of
\emph{Type~II} if any $x\in C$ has weight divisible by $4$.  Otherwise $C$ is of \emph{Type~I}.
\end{definition}

\begin{lemma}
\label{lem:codes-lagrangians}
Every self-dual code $C\subseteq\F_2^d$ satisfies
$R\subseteq C\subseteq E$, and $L=C/R$ is a Lagrangian subspace of $W$.
Conversely, the inverse image in $E$ of every Lagrangian $L\subseteq W$ is a
self-dual code.

On a Lagrangian, $q|_L$ is linear.  The code $C$ is Type II exactly when
$q|_L=0$.  If $C$ is Type I, then exactly half of its vectors have weight
congruent to $2$ modulo $4$.
\end{lemma}

\begin{proof}
If $C=C^\perp$, by definition, every $c\in C$ has even
weight and $C\subseteq E$.  The all-one vector is orthogonal to every even
vector, so $\one\in C^\perp=C$ and $R\subseteq C$.  The image $L=C/R$ is
isotropic and
\[
       \dim L=\dim C-1=\frac d2-1=h;
\]
therefore it is Lagrangian.  Conversely, the inverse image of a Lagrangian
has dimension $d/2$ and is self-orthogonal, hence self-dual.

The polar form $b$ vanishes on $L$, so
$q(x+y)=q(x)+q(y)+b(x,y)=q(x)+q(y)$ for all $x,y\in L$.  Thus
$q|_L$ is linear.  By definition of $q$, one has
$q([x])=0$ if and only if $\wt(x)\equiv0\pmod4$.  Hence $C$ is Type~II if
and only if $q|_L\equiv0$.

Now suppose $C$ is Type~I.  Then $q|_L$ is not identically zero, so it is a
nonzero linear map $L\to\F_2$.  In particular there exists $z_0\in L$ with
$q(z_0)=1$.  The translation $z\mapsto z+z_0$ is a bijection of $L$ sending
$\{z\in L:q(z)=0\}$ onto $\{z\in L:q(z)=1\}$, and therefore
\[
 \bigl|\{z\in L:q(z)=1\}\bigr|
 =\frac12|L|.
\]
Each $z\in L=C/R$ lifts to exactly two vectors $x,x+\one\in C$, and
these two representatives satisfy $q([x])=q([x+\one])$.  Consequently exactly half of the vectors of $C$ have $q=1$.
\end{proof}

\subsection{Finite geometry incidence bounds}\label{subsec:finite-geometry-incidence-bounds}

We include some counting arguments about Type I and Type II Lagrangians, which will be used for our further analysis. 

\begin{lemma}
\label{lem:all-lagrangians}
Let $(W,b)$ be a $2h$-dimensional symplectic space over $\F_2$, and let
$A_h$ denote the number of Lagrangians in $W$.  Then
\[
                         A_h=\prod_{j=1}^h(2^j+1).
\]
Moreover, for a uniform Lagrangian
$L\subseteq W$ and every fixed nonzero vector $z\in W$,
\[
                         \Prb[z\in L]=\frac1{2^h+1}.
\]
\end{lemma}

\begin{proof}
We use the double counting argument from \cite{Taylor1992}.
Double-count pairs $(L,z)$ with $L$ Lagrangian and $0\neq z\in L$.
Every Lagrangian has dimension $h$, hence contains $2^h-1$ nonzero vectors.
Therefore the number of such pairs is $A_h(2^h-1)$.

Now fix $0\neq z\in V$.  If $L$ is a Lagrangian containing $z$, then
$b(z,x)=0$ for every $x\in L$, so $L\subseteq z^\perp$.  Combined with
$\langle z\rangle\subseteq L$, we obtain
\[
 \langle z\rangle\subseteq L\subseteq z^\perp.
\]
Here $\dim z^\perp=2h-1$, so the quotient space $z^\perp/\langle z\rangle$ has
dimension $2h-2$ and is again a symplectic space.  The correspondence
\[
 L\;\longmapsto\; L/\langle z\rangle
\]
is a bijection between Lagrangians of $V$ containing $z$ and Lagrangians of
$z^\perp/\langle z\rangle$; each such Lagrangian has dimension $h-1$.
Hence there are exactly $A_{h-1}$ Lagrangians containing the fixed $z$.

On the other hand, $V$ has $2^{2h}-1$ nonzero vectors.  Therefore we have
\[
 A_h(2^h-1)=(2^{2h}-1)A_{h-1},\qquad A_0=1.
\]
Since $(2^{2h}-1)/(2^h-1)=2^h+1$, iteration gives
$A_h=\prod_{j=1}^h(2^j+1)$.  For a uniformly random Lagrangian $L$, the
probability that it contains a fixed nonzero $z$ is
$A_{h-1}/A_h=1/(2^h+1)$.
\end{proof}

\begin{lemma}
\label{lem:type-two-count}
Let $B_h$ denote the number of Lagrangians $L=C/R$ in $W$ arising from
Type~II self-dual codes $C\subseteq\F_2^d$.  Then
\[
                         B_h=2\prod_{j=1}^{h-1}(2^j+1).
\]
Consequently a uniform Lagrangian of $W$ arises from a Type~II code with
probability $2/(2^h+1)$.  Moreover, if $L$ is a uniformly random Lagrangian
arising from a Type~II code and $0\neq z\in W$, then
\[
 \Prb[z\in L]=
 \begin{cases}
 0,&q(z)=1,\\[1mm]
 \dfrac1{2^{h-1}+1},&q(z)=0.
 \end{cases}
\]
\end{lemma}

\begin{proof}
By \Cref{lem:codes-lagrangians}, these are precisely the $h$-dimensional
totally singular subspaces of $(W,q)$. We use the double counting argument from \cite{Taylor1992}. 
We now double-count pairs $(L,z)$ where $L$ is an $h$-dimensional totally
singular subspace and $0\neq z\in L$.  Counting first by $L$ gives
$B_h(2^h-1)$.  For a fixed singular $z$, such subspaces correspond, via
$L\mapsto L/\langle z\rangle$, to the $(h-1)$-dimensional totally singular
subspaces of $z^\perp/\langle z\rangle$.  Indeed, $q(x+z)=q(x)$ for
$x\in z^\perp$, so $q$ descends to this quotient, and its polar form is
nondegenerate because the radical of $b|_{z^\perp}$ is $\langle z\rangle$.
The quotient also contains an $(h-1)$-dimensional totally singular subspace:
use $L_0/\langle z\rangle$ if $z\in L_0$, and the image of
$L_0\cap z^\perp$ otherwise.  Hence there are $B_{h-1}$ such subspaces. We let $S_h$ be the number of nonzero vectors $x$ satisfying $q(x)=0$. Double counting therefore gives, 
\[
 B_h(2^h-1)=S_hB_{h-1},\qquad B_0=1.
\]

Next we compute $S_h$. Let $L_0$ be the $h$-dimensional
totally singular subspace supplied by \Cref{lem:quadratic-quotient}, and choose
a basis $e_1,\ldots,e_h$ of $L_0$.  Since $L_0$ is Lagrangian one has
$L_0^\perp=L_0$.  Nondegeneracy of $b$ therefore induces an isomorphism
$W/L_0\cong L_0^*$ via $y+L_0\mapsto f_y$ with linear $f_y: x\mapsto b(x,y)$.  Taking
a dual basis, we may choose $f_1,\ldots,f_h\in W$ with
$b(e_i,f_j)=\delta_{ij}$.  Replacing each $f_i$ by a vector in the coset
$f_i+L_0$ we may further assume $b(f_i,f_j)=0$.  Thus
\[
 b(e_i,f_j)=\delta_{ij},\qquad b(f_i,f_j)=0.
\]
Replacing $f_i$ by $f_i+q(f_i)e_i$ makes $q(f_i)=0$ without changing these
identities.  The $2h$ vectors $\{e_1,\ldots,e_h,f_1,\ldots,f_h\}$ are linearly
independent, hence a basis of $W$.  Every $x\in W$ therefore has a unique
expression $\sum_i(a_i e_i+b_i f_i)$, and
\[
 q\!\left(\sum_i a_i e_i+b_i f_i\right)=\sum_i a_i b_i.
\]
There are $2^h$ choices of $b$ when $a=0$, and $2^{h-1}$ when $a\neq0$.
Thus the number of nonzero vectors $x$ satisfying $q(x)=0$ is
\[
 S_h=2^h+(2^h-1)2^{h-1}-1
     =(2^h-1)(2^{h-1}+1).
\]
Therefore
\[
 B_h=2\prod_{j=1}^{h-1}(2^j+1).
\]
Since $A_h=\prod_{j=1}^h(2^j+1)$, we have
$B_h/A_h=2/(2^h+1)$.  Finally, a vector with $q(z)=1$ lies in no such $L$,
whereas for $0\neq z$ with $q(z)=0$ the same correspondence gives
\[
 \Prb[z\in L]=\frac{B_{h-1}}{B_h}
             =\frac1{2^{h-1}+1}.
\]
\end{proof}

We use the following two ensembles of codes:
\begin{itemize}
 \item $\mathcal C_{\mathrm{II}}$, the uniform Type-II self-dual ensemble;
 \item $\mathcal C_{\mathrm I}^{\mathrm{good}}$, the uniform Type-I
       self-dual ensemble conditioned on having no weight-two code.
\end{itemize}

\begin{lemma}[The good Type-I ensemble]
\label{lem:good-type-one}
For a uniform Type-I Lagrangian $L$, the probability that it contains a weight-two vector is at most $\frac{\binom d2}{2^h-1}\le \frac{1}{2}$ for all sufficiently large $d$. In addition, 
for each fixed $0\neq z\in W$, the following holds:
\[
 \Prb_{C\sim\mathcal C_{\mathrm I}^{\mathrm{good}}}[z\in L=C/R]
        \leq\frac2{2^h-1}\leq2^{2-h}.
\]
\end{lemma}

\begin{proof}
By \Cref{lem:all-lagrangians,lem:type-two-count}, a uniform Lagrangian is
Type I with probability
\[
 1-\frac2{2^h+1}=\frac{2^h-1}{2^h+1}.
\]
Thus, for fixed $z\neq0$,
\[
 \Prb[z\in L=C/R\mid C\text{ is Type I}]
 \leq \frac{1/(2^h+1)}{(2^h-1)/(2^h+1)}
 =\frac1{2^h-1}.
\]
A union bound over the $\binom d2$ weight-two vectors
gives the probability upper bound $\frac{\binom d2}{2^h-1}$.  For large $d$, this is at most $1/2$. So $\Prb_{C\sim\mathcal C_{\mathrm I}^{\mathrm{good}}}[z\in L=C/R]
\leq\frac2{2^h-1}\leq2^{2-h}$.
\end{proof}

\subsection{Construction and property of the hard instance.}

We now construct the hard instance by Cayley graph construction.

\begin{definition}
\label{def:cayley-graph}
Let $\Gamma$ be a finite abelian group, and let
$S\subseteq\Gamma\setminus\{0\}$ satisfy $-S=S$.  The \emph{Cayley graph} $\Cay(\Gamma,S)$ has vertex set $\Gamma$, and for two vertices $x,y\in \Gamma$, $(x,y)$ is an edge whenever $y-x\in S$. We call $S$ \emph{generators} of $\Cay(\Gamma,S)$.
\end{definition}

For a self-dual code $C\subseteq\F_2^d$, define
\[
 G_C=\Cay\!\left(\F_2^d/C,
       \{e_1+C,\ldots,e_d+C\}\right).
\]

\Cref{fig:cayley-quotient} illustrates the construction on the three-dimensional cube.

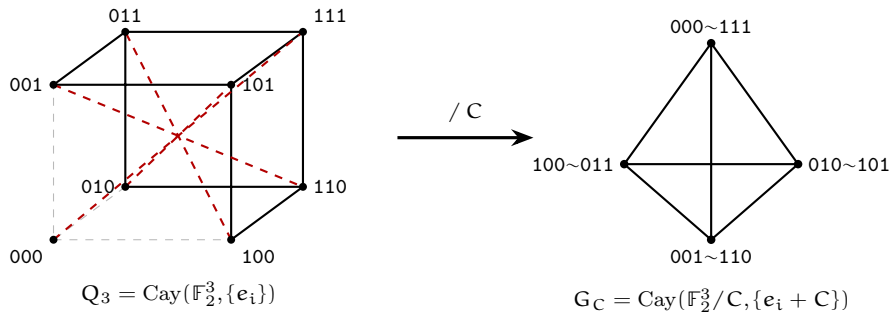
\begin{figure}[ht]
\centering
\begin{tikzpicture}[font=\scriptsize, >=Stealth]
  \begin{scope}
    \coordinate (v000) at (0.00,0.00);
    \coordinate (v100) at (2.35,0.00);
    \coordinate (v010) at (0.95,0.70);
    \coordinate (v110) at (3.30,0.70);
    \coordinate (v001) at (0.00,2.05);
    \coordinate (v101) at (2.35,2.05);
    \coordinate (v011) at (0.95,2.75);
    \coordinate (v111) at (3.30,2.75);
    \draw[gray!55, dashed] (v010) -- (v000) -- (v001);
    \draw[gray!55, dashed] (v000) -- (v100);
    \draw[thick] (v100) -- (v110) -- (v010) -- (v011) -- (v001)
                 -- (v101) -- (v100);
    \draw[thick] (v101) -- (v111) -- (v011);
    \draw[thick] (v110) -- (v111);
    \draw[red!70!black, dashed, thick] (v000) -- (v111);
    \draw[red!70!black, dashed, thick] (v100) -- (v011);
    \draw[red!70!black, dashed, thick] (v010) -- (v101);
    \draw[red!70!black, dashed, thick] (v001) -- (v110);
    \fill (v000) circle (1.5pt); \node[below left]  at (v000) {\texttt{000}};
    \fill (v100) circle (1.5pt); \node[below right] at (v100) {\texttt{100}};
    \fill (v010) circle (1.5pt); \node[left]        at (v010) {\texttt{010}};
    \fill (v110) circle (1.5pt); \node[right]       at (v110) {\texttt{110}};
    \fill (v001) circle (1.5pt); \node[left]        at (v001) {\texttt{001}};
    \fill (v101) circle (1.5pt); \node[right]       at (v101) {\texttt{101}};
    \fill (v011) circle (1.5pt); \node[above]       at (v011) {\texttt{011}};
    \fill (v111) circle (1.5pt); \node[above right] at (v111) {\texttt{111}};
    \node at (1.65,-0.72) {$Q_3=\Cay(\F_2^3,\{e_i\})$};
  \end{scope}

  \draw[very thick, ->] (4.55,1.35) -- (6.35,1.35);
  \node[above] at (5.45,1.40) {$/\,C$};

  \begin{scope}[shift={(7.55,0.15)}]
    \coordinate (a) at (1.15,2.45);
    \coordinate (b) at (0.00,0.85);
    \coordinate (c) at (2.30,0.85);
    \coordinate (d) at (1.15,-0.15);
    \draw[thick] (a) -- (b) -- (c) -- (a);
    \draw[thick] (a) -- (d) -- (b);
    \draw[thick] (c) -- (d);
    \foreach \p in {a,b,c,d} {\fill (\p) circle (1.6pt);}
    \node[above] at (a) {\texttt{000}$\sim$\texttt{111}};
    \node[left]  at (b) {\texttt{100}$\sim$\texttt{011}};
    \node[right] at (c) {\texttt{010}$\sim$\texttt{101}};
    \node[below] at (d) {\texttt{001}$\sim$\texttt{110}};
    \node at (1.15,-0.95) {$G_C=\Cay(\F_2^3/C,\{e_i+C\})$};
  \end{scope}
\end{tikzpicture}
\caption{Illustration of our construction with $C=\langle\mathbf{1}\rangle$ in the three-dimensional cube.}
\label{fig:cayley-quotient}
\end{figure}

\begin{lemma}
\label{lem:quotient-graph}
For every $C\in C_{\mathrm I}^{\mathrm{good}}$ and every $C\in\mathcal C_{\mathrm{II}}$, $G_C$ is a connected
simple unweighted $d$-regular graph on $n=2^{d/2}$ vertices.
\end{lemma}

\begin{proof}
The quotient has order $2^{d-\dim C}=2^{d/2}$.  There are no self loops because a
self-dual code is even and therefore $x+C\neq x+e_i+C$ for all $x\in \F_2^d$ and $i\in [d]$.  For $i\neq j$, the
generators $e_i+C=e_j+C$ exactly when $e_i+e_j\in C$, which can not happen for both $C\in \mathcal C_{\mathrm I}^{\mathrm{good}}$ and $C\in \mathcal C_{\mathrm{II}}$. This is because $\wt(x)\neq 2$ for all $x\in C$ whenever $C\in \mathcal C_{\mathrm I}^{\mathrm{good}}$ and $C\in \mathcal C_{\mathrm{II}}$.  Thus the generators are
nonzero and distinct, making the graph simple and $d$-regular.  They generate
the quotient, so the graph is connected.  
\end{proof}

Next we provide the exact spectral information of the constructed graph.

\begin{lemma}
\label{lem:quotient-spectrum}
For a self-dual code $C$, the normalized-adjacency eigenvalues of $G_C$ (with
multiplicity) are
\[
                   \lambda_y=1-\frac{2\wt(y)}d,
                   \qquad y\in C.
\]
\end{lemma}

\begin{proof}
Write $V(G_C)=\F_2^d/C$ and let $A$ be the adjacency matrix of $G_C$.
For each $y\in C$ define $f_y:V(G_C)\to\R$ by
\[
 f_y(x+C)=(-1)^{x\cdot y}.
\]
This is well defined: if $x'=x+c$ with $c\in C$, then
$x'\cdot y=x\cdot y+c\cdot y=x\cdot y$ because $y\in C=C^\perp$.
Applying $A$ at the vertex $x+C$ gives
\begin{align*}
 (Af_y)(x+C)
 &=\sum_{i=1}^d f_y(x+e_i+C)
 =\sum_{i=1}^d (-1)^{(x+e_i)\cdot y}\\
 &=(-1)^{x\cdot y}\sum_{i=1}^d (-1)^{y_i}
 =\Bigl(\sum_{i=1}^d (-1)^{y_i}\Bigr)f_y(x+C).
\end{align*}
It satisfies $\sum_{i=1}^d (-1)^{y_i}=d-2\wt(y)$, so $f_y$ is an eigenvector of $A$ with
eigenvalue $d-2\wt(y)$.  The normalized adjacency is $N=A/d$, hence
\[
 Nf_y=\Bigl(1-\frac{2\wt(y)}d\Bigr)f_y.
\]
There are $|C|=2^{d/2}=n$ such vectors $\{f_y\}_{y\in C}$.  They are
pairwise orthogonal in $\R^{V(G_C)}$: if $y\neq y'$ in $C$, set
$z=y+y'\neq0$.  Then $z\in C=C^\perp$, so $(-1)^{x\cdot z}$ is constant on
each coset of $C$.  Moreover, it holds that
\[
 \sum_{x\in\F_2^d}(-1)^{x\cdot z}=0.
\]
To see this, by $z\neq0$, there is a coordinate $j$ with $z_j=1$, and pairing
each $x$ with the vector $x'$ obtained by flipping the $j$-th bit gives
$(-1)^{x'\cdot z}=(-1)^{x\cdot z+1}=-(-1)^{x\cdot z}$, so the sum cancels
in pairs.  Grouping by cosets therefore yields
$\sum_{x+C}f_y(x+C)f_{y'}(x+C)=0$.  Hence $\{f_y\}_{y\in C}$ is an
eigenbasis of $N$, and the eigenvalues $1-\frac{2\wt(y)}d$ are the full spectrum with
multiplicity.
\end{proof}

Next we prove the spectral density of $G_C$ for $C\in \mathcal C_{\mathrm{II}}$ and $C\in \mathcal C_{\mathrm I}$ seperates. We refer to \Cref{fig:type-spectra} for illustration of the spectral densities of the two type of graphs. Define the lattice
\[
 \Lambda_d=\left\{1-\frac{8j}{d}:0\leq j\leq\frac d4\right\}
\]
and its distance function
\[
                 f_d(x)=\dist(x,\Lambda_d),\qquad x\in[-1,1].
\]
Clearly, the function $f_d$ is $1$-Lipschitz.

\begin{proposition}
\label{prop:classical-separation}
For every $C_{\mathrm{II}}\in\mathcal C_{\mathrm{II}}$ and every
$C_{\mathrm I}\in\mathcal C_{\mathrm I}$,
\[
 \int f_d\,d\rho_{G_{C_{\mathrm{II}}}}=0,
 \qquad
 \int f_d\,d\rho_{G_{C_{\mathrm I}}}=\frac2d.
\]
Consequently
\[
 W_1\!\left(\rho_{G_{C_{\mathrm{II}}}},
             \rho_{G_{C_{\mathrm I}}}\right)\geq\frac2d.
\]
\end{proposition}

\begin{proof}
For Type II, every codeword has weight $0$ modulo $4$, so
\Cref{lem:quotient-spectrum} places every eigenvalue in $\Lambda_d$.
For Type I, exactly half of the $|C|=2^{d/2}$ codewords have weight
$2$ modulo $4$.  If $\wt(y)=4j+2$, then
\[
                 \lambda_y=1-\frac{8j+4}{d},
\]
whose distance from $\Lambda_d$ is exactly $4/d$.  The remaining half have
weight $0$ modulo $4$ and contribute zero.  Since the empirical spectral
measure assigns mass $1/|C|$ to each character/eigenvalue, the integral is
\[
             \frac{|C|/2}{|C|}\cdot\frac4d=\frac2d.
\]
By \Cref{prop:kr-duality}, we have $W_1\!\left(\rho_{G_{C_{\mathrm{II}}}},
             \rho_{G_{C_{\mathrm I}}}\right)\geq\frac2d$.
\end{proof}

\begin{figure}[ht]
  \centering
  \begin{tikzpicture}
  \begin{axis}[
    name=typeII,
    width=0.46\textwidth,
    height=3.4cm,
    xmin=-1.18, xmax=1.18,
    ymin=0, ymax=0.86,
    xtick={-1,-0.5,0,0.5,1},
    extra x ticks={-0.75,-0.25,0.25,0.75},
    extra x tick labels={},
    extra x tick style={major tick length=2.5pt, tick style={gray}},
    ytick={0,0.2,0.4,0.6,0.8},
    xlabel={eigenvalue of $N_G$},
    title={Type~II},
    title style={font=\small\bfseries},
    tick label style={font=\small},
    tick align=inside,
    axis background/.style={fill=white},
    clip=false,
  ]
  \pgfplotsinvokeforeach{-1,-0.5,0,0.5,1}{
    \draw[gray!45, densely dotted] (axis cs:#1,0) -- (axis cs:#1,0.86);
  }
  \addplot[
    ycomb,
    blue!70!black,
    line width=0.9pt,
    mark=*,
    mark size=1.5pt,
    mark options={draw=red, fill=red},
  ] coordinates {
    (-1, 0.0039)
    (-0.5, 0.1094)
    (0, 0.7734)
    (0.5, 0.1094)
    (1, 0.0039)
  };
  \end{axis}
  
  \begin{axis}[
    name=typeI,
    at={(typeII.east)},
    anchor=west,
    xshift=0.85cm,
    width=0.46\textwidth,
    height=3.4cm,
    xmin=-1.18, xmax=1.18,
    ymin=0, ymax=0.86,
    xtick={-1,-0.5,0,0.5,1},
    extra x ticks={-0.75,-0.25,0.25,0.75},
    extra x tick labels={},
    extra x tick style={major tick length=2.5pt, tick style={gray}},
    ytick={0,0.2,0.4,0.6,0.8},
    xlabel={eigenvalue of $N_G$},
    title={good Type~I},
    title style={font=\small\bfseries},
    tick label style={font=\small},
    tick align=inside,
    axis background/.style={fill=white},
    clip=false,
  ]
  \pgfplotsinvokeforeach{-1,-0.5,0,0.5,1}{
    \draw[gray!45, densely dotted] (axis cs:#1,0) -- (axis cs:#1,0.86);
  }
  \addplot[
    ycomb,
    blue!70!black,
    line width=0.9pt,
    mark=*,
    mark size=1.5pt,
    mark options={draw=red, fill=red},
  ] coordinates {
    (-1, 0.0039)
    (-0.5, 0.0938)
    (-0.25, 0.2500)
    (0, 0.3047)
    (0.25, 0.2500)
    (0.5, 0.0938)
    (1, 0.0039)
  };
  \end{axis}
  \end{tikzpicture}
  \caption{Spectral densities of the two type of graphs at $d=16$, $n=2^{d/2}=256$.}
  \label{fig:type-spectra}
  \end{figure}
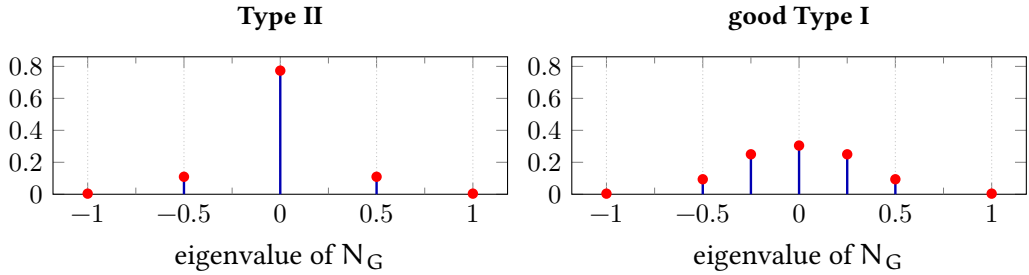

\subsection{Indistinguishability via coupling.}\label{subsec:indistinguishability-via-coupling}

We prove indistinguishability under local access.  Define the shadow graph
\[
 U=\F_2^d/R,
 \qquad
 H=\Cay\!\left(U,\{e_1+R,\ldots,e_d+R\}\right),
\]
and, for $L=C/R$, the quotient map
\[
 \pi_L:U\longrightarrow U/L=\F_2^d/C,
\]
which sends the uniform distribution on $U$ to the uniform distribution on
$\F_2^d/C$.  Next we prove that the local
access queries cannot be distinguished via coupling. Our argument is based
on Yao's minimax principle~\cite{Yao1977}, so we first fix any deterministic
$Q$-query algorithm and two input distributions: (i) $G_{C_I}$ with a uniform draw
$C_{\mathrm I}\sim\mathcal C_{\mathrm I}^{\mathrm{good}}$ and (ii) $G_{C_{II}}$ with a uniform
draw $C_{\mathrm{II}}\sim\mathcal C_{\mathrm{II}}$. We prove that the total variation distance between the query transcripts of the two graphs is exponentially small.

\begin{lemma}
\label{lem:adaptive-coupling}
Let $\mathcal A$ be an arbitrary deterministic algorithm that makes at most
$Q$ local-access queries with $Q\le 2^h$.  Draw $C_{\mathrm I}\sim\mathcal C_{\mathrm I}^{\mathrm{good}}$
and $C_{\mathrm{II}}\sim\mathcal C_{\mathrm{II}}$ independently and uniformly,
and write $T_{\mathrm I}$ (resp.\ $T_{\mathrm{II}}$) for the complete query
transcript of $\mathcal A$ on $G_{C_{\mathrm I}}$ (resp.\ $G_{C_{\mathrm{II}}}$). Then
\[
                 \TV(T_{\mathrm I},T_{\mathrm{II}})
                    \leq3(Q+1)^2 2^{-h}.
\]
\end{lemma}

\begin{proof}
Since $\mathcal A$ is deterministic, its next query is a fixed function of the
answers received so far.  We couple both input distributions to a single
execution of $\mathcal A$ against the shadow graph $H$, as follows.
Throughout, draw a uniformly random permutation $\sigma$ of
$[n]=\{1,\ldots,2^{d/2}\}$ and assign labels from this random ordering:
the $k$th previously unseen vertex receives $\sigma(k)$; a previously seen
vertex keeps the same label.  A coupling argument is detailed in the following.
\begin{itemize}
\item \textbf{Query transcript on $H$.}
We answer the queries of $\mathcal A$ on $H$ as follows: (i) a uniform-vertex query returns
an independent uniform $x\in U$; (ii) a neighbor query in direction $i$ at a
previously returned vertex $x$ returns $x+e_i+R$.  Present each answer to
$\mathcal A$ by the label assigned under the rule above. That is, the $k$-th
newly seen vertex $x_j$ receives $\sigma(k)$, and a previously seen vertex
receives the same label.  This produces a sequence of at most $Q+1$ shadow vertices
$x_0,\ldots,x_t\in U$ together with a shadow transcript $T_H$ of returned
labels. 

\item \textbf{Query transcript on $G_{C_{\mathrm I}}$ and $G_{C_{\mathrm{II}}}$.}
Given a hidden code $C\in\{C_{\mathrm I},C_{\mathrm{II}}\}$ with $L=C/R$,
push the same shadow vertex sequence forward under $\pi_L$ step by step. Specifically, at the initial step we query $\pi_L(x_0)$ on $G_C$. Assume at step $t\le Q-1$ the coupling keeps success and we have already query the vertices
\[
 \pi_L(x_0),\;\pi_L(x_1),\;\ldots,\;\pi_L(x_t)\in V(G_C).
\]
The labels to these vertices are assigned by the same permutation $\sigma$
as on $H$: the $k$-th newly seen image receives $\sigma(k)$. We obtain the current $t$ step transcript, denoted as $T_C^{(t)}$. Denote the first $t$ elements of $T_H$ as $T_H^{(t)}$. We now query the next vertex $x_{t+1}$ on $H$. Consider the following two cases: (i) If $T_H^{(t)}=T_C^{(t)}$, we query $\pi_L(x_{t+1})$ on $G_C$. This is because $\mathcal A$ is deterministic, and we have only explored the same structure on $H$ and $G_C$. The query choices therefore remain identical to the execution on $H$. We assign the label to $\pi_L(x_{t+1})$. (ii) If $T_H^{(t)}\neq T_C^{(t)}$, the coupling fails and we assign the label according to the $(t+1)$-th query on $G_C$. After $Q$ steps, we write $T_C$ for the resulting transcript of returned labels.
\end{itemize}

Next we consider the condition of disagreement. As long as $\pi_L$ is injective on the visited set $\{x_0,\ldots,x_t\}$, it holds that $T_H^{(t)}=T_C^{(t)}$ at each step $t\le Q$.
The first time $T_C$ can differ from $T_H$ is when two distinct visited
vertices $x_i,x_j$ on $H$ satisfy $x_i\neq x_j$ but $\pi_L(x_i)=\pi_L(x_j)$
for some $i,j\in[Q]$.  This holds if and only if $z:=x_i-x_j\in L$. 

We next bound the disagreement probability.  For a pair with $z\notin W$ the
collision probability is zero because $L\subseteq W$ by \Cref{lem:codes-lagrangians}.  For a pair with $0\neq z\in W$,
\Cref{lem:good-type-one} gives $\Prb[z\in L=C_I/R]\le 2^{2-h}$ under
$C_I\sim\mathcal C_{\mathrm I}^{\mathrm{good}}$, and
\Cref{lem:type-two-count} gives $\Prb[z\in L=C_{II}/R]\le 2^{1-h}$ under
$C_{II}\sim\mathcal C_{\mathrm{II}}$. 

Since there are at most $Q+1$ vertices appear on $H$, fewer than $(Q+1)^2/2$ unordered
pairs are tested.  A union bound yields
\begin{align*}
 \TV(T_{\mathrm{II}},T_H)
    &\le \Prb[T_H\neq T_{C_{\mathrm{II}}}] \le \frac{(Q+1)^2}{2}\,2^{1-h}
      =(Q+1)^2 2^{-h},\\
 \TV(T_{\mathrm I},T_H)
    &\le \Prb[T_H\neq T_{C_{\mathrm I}}] \le \frac{(Q+1)^2}{2}\,2^{2-h}
      =2(Q+1)^2 2^{-h}.
\end{align*}
The triangle inequality gives
$\TV(T_{\mathrm I},T_{\mathrm{II}})\leq3(Q+1)^2 2^{-h}$.

\end{proof}

Finally, we put things together to prove the exponential lower bound.

\begin{proof}[Proof of \Cref{thm:intro-classical}]
By Yao's minimax principle~\cite{Yao1977}, it suffices to prove the lower
bound against deterministic algorithms on random problem inputs.  Fix an arbitrary deterministic
algorithm $\mathcal A$ that makes at most $Q$ local-access queries.  Draw an
input graph by the following distribution: with probability $1/2$ sample
$C\sim\mathcal C_{\mathrm I}^{\mathrm{good}}$ uniformly, and with probability
$1/2$ sample $C\sim\mathcal C_{\mathrm{II}}$ uniformly; then take $G=G_C$.

Suppose $\mathcal A$ outputs a spectral density estimate $\widehat\rho$ with
Wasserstein error at most $\eps<1/d$ with probability at least $2/3$ on the distribution of problem inputs.  Given its output, declare $C\sim\mathcal C_{\mathrm{II}}$ when
\[
                         \int f_d\,d\widehat\rho<\frac1d
\]
and $C\sim\mathcal C_{\mathrm I}^{\mathrm{good}}$ otherwise.  By
\Cref{prop:classical-separation}, this test is correct whenever the
estimator succeeds.  Hence it distinguishes the two transcript laws
$T_{\mathrm I}$ and $T_{\mathrm{II}}$ with success probability at least
$2/3$, which forces
\[
 \TV(T_{\mathrm I},T_{\mathrm{II}})\geq\frac13.
\]
By \Cref{lem:adaptive-coupling} we therefore have
\[
           3(Q+1)^2 2^{-h}\geq\frac13,
           \qquad\text{so}\qquad
           Q+1\geq\frac13\,2^{h/2}.
\]
As $h=d/2-1$, so $Q=2^{\Omega(d)}$.
For sufficiently small $\eps$, choose
\[
                       d=8\left\lfloor\frac1{16\eps}\right\rfloor.
\]
After decreasing a universal $\eps_0$ if necessary, we guarantee that
\[
 \frac1{4\eps}\leq d\leq\frac1{2\eps},
 \qquad\text{so}\qquad
 h=\frac{d}{2}-1=\Theta(1/\eps).
\]
Thus $Q\geq2^{c/\eps}$ for a universal $c>0$.  This finishes the proof.
\end{proof}

\section{Polynomial-query quantum algorithms}
\label{sec:quantum-upper}

We next prove \Cref{thm:intro-quantum-upper}.  The
first step removes matrix entries incident to very high degree.  The second
step constructs a projected-unitary encoding of the truncated normalized
adjacency matrix.  The final step we learns the resulting spectral density in the Wasserstein metric by utilizing the quantum
$\ell_\infty$-norm multidimensional amplitude estimation result \cite{vanApeldoorn2021}.

\subsection{High-degree spectral truncation}

In this subsection, we truncate the high-degree vertices of the graph, which is inspired by the nuclear sparcification method from \cite{JinEtAl2024}. Fix a threshold $\Delta\geq1$, let
\[
 H=\{v:d_v>\Delta\},\qquad L=V\setminus H,
\]
and let $P_L$ be the diagonal projector onto $L$.  Define
\[
                              B=P_LNP_L.
\]
We remind that $B$ is not the normalized adjacency of the induced subgraph $G[L]$. However, the spectral density $\rho_B$ is a good approximation of $\rho_N$. The following Lemma is the similar argument as in \cite{JinEtAl2024}, we provide the proof here for completeness.

\begin{lemma}
\label{lem:truncation}
For every graph and every $\Delta\geq1$,
\[
                         W_1(\rho_N,\rho_B)
                              \leq\sqrt{\frac2\Delta}.
\]
\end{lemma}

\begin{proof}
Let $M=N-B$.  Its nonzero entries correspond precisely to edges incident to
at least one high-degree vertex.  For low--high edges,
\[
 \sum_{\substack{\{u,v\}\in E\\u\in L,\ v\in H}}
       \frac1{d_ud_v}
 \leq \sum_{u\in L}\frac1{d_u}\cdot
       \frac{|N(u)\cap H|}{\Delta}
 \leq\frac{|L|}{\Delta}.
\]
These entries occur in both orientations in the Frobenius norm.  For the
oriented high--high entries,
\[
 \sum_{u\in H}\frac1{d_u}
       \sum_{\substack{v\in H\\v\sim u}}\frac1{d_v}
 \leq \sum_{u\in H}\frac1{d_u}\cdot\frac{d_u}{\Delta}
 =\frac{|H|}{\Delta}.
\]
Consequently
\[
                 \norm{N-B}_F^2
                    \leq\frac{2|L|+|H|}{\Delta}
                    \leq\frac{2n}{\Delta}.
\]
Hoffman--Wielandt's inequality~\cite{Bhatia1997}, followed by
Cauchy--Schwarz, gives
\[
 W_1(\rho_N,\rho_B)
 \leq\left(\frac1n\sum_{j=1}^n
       |\lambda_j(N)-\lambda_j(B)|^2\right)^{1/2}
 \leq\frac{\norm{N-B}_F}{\sqrt n}
 \leq\sqrt{\frac2\Delta}.
\]
\end{proof}

\subsection{Preparing the neighbor state.}

Call $v$ \emph{active} if $1\leq d_v\leq\Delta$, and for an active vertex
write
\[
                      \ket{p_v}=\frac1{\sqrt{d_v}}
                             \sum_{u\sim v}\ket u.
\]

Preparing $\ket{p_v}$ is the key step in our quantum algorithm design. The following Lemma shows that $\ket{p_v}$ can be prepared in $O(\sqrt\Delta)$ queries for active $v$. We utilize the Grover's search.

\begin{lemma}\cite{Grover1996,BrassardEtAl2002}
  \label{lem:grover}
  Let \(N\ge1\) and let \(f:[N]\to\{0,1\}\) have a unique marked element
  \(x_*\), so that \(f(x_*)=1\).  Suppose a unitary \(O_f\) implements
  \[
   O_f\ket x\ket b=\ket x\ket{b \oplus f(x)}.
  \]
  There is a unitary operation that uses \(O(\sqrt N)\) applications of \(O_f\) and
  \(O_f^\dagger\) and implements the exact map
  \[
   \ket0\longmapsto\ket{x_*}.
  \]
  The same bound holds for the controlled map
  \(\ket\psi\ket0\mapsto\ket\psi\ket{x_*}\).
  \end{lemma}

\begin{lemma}
\label{lem:port-erasure}
There is a unitary state-preparation procedure $T$ such that
\[
 T\ket v=
 \begin{cases}
  \ket0_f\ket0_b\ket v\ket{p_v},&1\leq d_v\leq\Delta,\\
  \ket1_f\ket0_b\ket v\ket\bot,&d_v=0\text{ or }d_v>\Delta,
 \end{cases}
\]
Where $f$ is the activity flag and $b$ is the auxiliary bit. In addition, $T$ uses $O(\sqrt\Delta)$ coherent degree and indexed-neighbor queries.
The inverse and controlled procedure have the same asymptotic cost.
\end{lemma}

\begin{proof}
Query $d_v$ coherently and compute the activity flag $f$. We set $f=1$ if $d_v=0$ or $d_v>\Delta$, and $f=0$ if $1\leq d_v\leq\Delta$. All subsequent
operations are controlled by the stored value of $d_v$, so the construction
also works when $v$ is in superposition.

Fix an active branch and write $d=d_v$, we set $f=0$.  Using controlled rotations, prepare a
uniform port and then make one neighbor query:
\[
 \ket0\ket0
 \longmapsto
 \frac1{\sqrt d}\sum_{i=1}^d\ket i\ket0
 \longmapsto
 \frac1{\sqrt d}\sum_{i=1}^d\ket i\ket{u_i},
 \qquad u_i=\operatorname{nbr}(v,i).
\]
We note that the uniform-state preparation uses no graph query.

It remains to erase the port $i$.  For every promised adjacent pair $(v,u)$,
let $\operatorname{port}_v(u)$ be the unique $i\in[d]$ such that
$\operatorname{nbr}(v,i)=u$.  The solution is unique because the graph is
simple. Using \Cref{lem:grover}, we obtain the following transformation with $O(\sqrt d)$ adjacency queries:
\[
 \ket v\ket u\ket0
 \longmapsto
 \ket v\ket u\ket{\operatorname{port}_v(u)}
\]

After the neighbor query the state is $\frac1{\sqrt d}\sum_i\ket i\ket{u_i}$.
 Rearranging the last two registers
gives $\frac1{\sqrt d}\sum_i\ket{u_i}\ket i$. Since $\operatorname{port}_v(u_i)=i$, the inverse of the map from
\Cref{lem:grover} therefore sends the state to:
\[
  \frac1{\sqrt d}\sum_{i=1}^d\ket{u_i}\ket i=\frac1{\sqrt d}\sum_{i=1}^d\ket{u_i}\ket{\operatorname{port}_v(u_i)}\longmapsto \frac1{\sqrt d}\sum_{i=1}^d\ket{u_i}\ket0
 =\ket{p_v}\ket0.
\]
Hence the port register is restored to \(\ket0\). 

On an inactive branch, set $f=1$ and prepare the fixed dummy state
$\ket\bot$.  The query complexity of both cases are bounded by $O(\sqrt{d})=O(\sqrt\Delta)$ queries.  Reversing the circuit implements
$T^\dagger$ at the same cost, and controlling the elementary operations
changes the cost by only a constant factor.
\end{proof}

\subsection{Estimating spectrum via phase estimation.}

Let $\Pi=TT^\dagger$ be the range projector.  On the target space of $T$,
define a unitary operator $U$ as follows: 
\begin{align*}
U\ket0_f\ket0_b\ket v\ket{u}&=\ket0_f\ket0_b\ket u\ket{v},\\
U\ket1_f\ket0_b\ket v\ket{\bot}&=\ket1_f\ket1_b\ket v\ket{\bot},\\
U\ket1_f\ket1_b\ket v\ket{\bot}&=\ket1_f\ket0_b\ket v\ket{\bot}.
\end{align*}

We next show that spectrum of $B$ can be estimated by quantum phase estimation.

\begin{definition}
      \label{def:cyclic-subspace}
      Let $V$ be a finite-dimensional vector space, let $A:V\to V$ be a
      linear operator, and let $v\in V$.  The \emph{cyclic subspace generated by
      $v$ under $A$} is the smallest $A$-invariant subspace containing $v$, namely
      \[
       \operatorname{span}\bigl\{v,\,Av,\,A^2v,\,\ldots\bigr\}
       =\operatorname{span}\bigl\{A^kv:k\geq0\bigr\}.
      \]
\end{definition}

\begin{definition}
\label{def:eigenphase}
Let $Q$ be a unitary operator.  If $Q\ket\phi=e^{i\theta}\ket\phi$ for a
unit vector $\ket\phi$ and $\theta\in\R$, we call $\theta$ an
\emph{eigenphase} of $Q$.  Eigenphases are identified modulo $2\pi$.
\end{definition}

\begin{lemma}
\label{lem:projected-encoding}
The compression of $U$ to the range of $T$ is exactly the truncated matrix:
\[
                              T^\dagger U T=B.
\]
Let the unitary operator
\[
                              Q=(2\Pi-I)U.
\]
For every eigenvalue $\lambda\in[-1,1]$ of $B$ with eigenvector $\ket\psi$, the cyclic subspace generated
from $T\ket\psi$ under $Q$ has eigenphases $\pm\arccos\lambda$.
Moreover, a controlled application of $Q$ or $Q^\dagger$ uses
$O(\sqrt\Delta)$ graph queries.
\end{lemma}

\begin{proof}
For active vertices $v,w$,
\begin{align*}
 \bra vT^\dagger UT\ket w
 &=\frac1{\sqrt{d_vd_w}}
   \sum_{u\sim v}\sum_{z\sim w}
       \langle v,u\mid z,w\rangle\\
 &=\frac{A_{vw}}{\sqrt{d_vd_w}}.
\end{align*}
If exactly one input is inactive, the flags are orthogonal, so $\bra vT^\dagger UT\ket w=0$.  If both are
inactive, $U$ flips $b$, so $\bra vT^\dagger UT\ket w=0$.  This
proves $T^\dagger UT=B$.

Now let $B\ket\psi=\lambda\ket\psi$, set
$\ket a=T\ket\psi$, and $\ket b=U\ket a$.  Since
$\Pi\ket b=\lambda\ket a$,
\[
 Q\ket a=2\lambda\ket a-\ket b,
 \qquad Q\ket b=\ket a.
\]
The characteristic polynomial on their span is
$z^2-2\lambda z+1$, with roots $e^{\pm i\arccos\lambda}$. So the eigenphases are $\pm\arccos\lambda$.

To implement the reflection $2\Pi-I$, apply the inverse of the state
preparation from \Cref{lem:port-erasure}, reflect about its zero-workspace
input subspace, and apply the state preparation again.  The middle
reflection and $U$ use no graph query.  Thus a controlled $Q$, and likewise $Q^\dagger$, costs $O(\sqrt\Delta)$ graph queries by
\Cref{lem:port-erasure}.
\end{proof}


Next we show that the eigenvalues of $B$ can be estimated via phase estimation applied to $Q$. 
For angles $\alpha,\beta\in\R$, write
\[
 d_{2\pi}(\alpha,\beta)
   :=\min_{k\in\mathbb Z}|\alpha-\beta+2\pi k|
\]
for their circular distance. We utilize the quantum phase estimation.

\begin{lemma}\cite{CleveEtAl1998,MandeDeWolf2023}
\label{lem:coherent-phase-estimation}
Let $R$ be a unitary operator, constants $0<\xi\leq 1/4$, 
$0<\delta\leq 1/2$.  There is a unitary procedure
$\operatorname{PE}_{\xi,\delta}(R)$ such that the following holds.  If
\[
                         R\ket{\phi_\theta}
                            =e^{i\theta}\ket{\phi_\theta},
\]
then measuring only the output register gives an estimate $\widehat\theta$
satisfying
\[
 \Pr\!\left[d_{2\pi}(\widehat\theta,\theta)>\xi\right]\leq\delta.
\]
 Each of
$\operatorname{PE}_{\xi,\delta}(R)$ and its inverse uses
\[
                         O\!\left(\frac1\xi
                                  \log\frac1\delta\right)
\]
controlled applications of $R$ or $R^\dagger$, respectively.
\end{lemma}

We now turn to the estimation of the eigenvalues of $B$.
We say a process is a \emph{coherent experiment} if it is a unitary whose
output register, upon measurement, yields a classical random variable.

\begin{lemma}
\label{lem:purified-spectrum}
For any resolution parameter $0<\eta\leq1$, there is a coherent experiment
$U_{\rm spec}$ that outputs a random variable $Z$ by measurment whose density function $\nu$ is supported
on $2/\eta$ points and satisfies
\[
                         W_1(\nu,\rho_B)\leq\eta.
\]
One use of $U_{\rm spec}$ or $U_{\rm spec}^\dagger$ costs
$ \widetilde O\!\left(
      \eta^{-1}\sqrt{\Delta}\right)$
graph queries.  Their controlled versions have the same asymptotic cost.
\end{lemma}

\begin{proof}
Prepare the maximally entangled vertex state with reference register $R$ and vertex register $V$:
\[
                    \ket\Phi=\frac1{\sqrt n}
                              \sum_{v\in V}\ket v_R\ket v_V,
\]
and apply $T$ to the vertex register.
Let $\{\ket{\psi_j}\}_{j=1}^n$ be a real orthonormal eigenbasis of $B$.
Basis invariance of $\ket\Phi$ gives
\[
 (I\otimes T)\ket\Phi
   =\frac1{\sqrt n}\sum_{j=1}^n
       \ket{\psi_j}_R T\ket{\psi_j}_V.
\]

By \Cref{lem:projected-encoding}, every eigenphase appearing in the $j$-th
cyclic subspace $\{Q^k\cdot T\ket{\psi_j}:k\ge0\}$ satisfies $\theta=\pm\arccos\lambda_j(B)$.  Therefore the ideal phase estimation $\cos(\theta)=\lambda_j(B)$ samples every eigenvalue with probability $1/n$.  
Apply \Cref{lem:coherent-phase-estimation} with phase accuracy $\eta/4$ and
failure probability $\eta/4$, take the cosine, and round it to a grid of
mesh $\eta/2$.  On the success event, if $\theta=\pm\arccos\lambda_j(B)$ is the eigenphase
in the $j$-th cyclic subspace, then the estimation $\widehat\theta$ by \ref{lem:coherent-phase-estimation} satisfies
\[
 \left|\operatorname{round}(\cos\widehat\theta)-\lambda_j(B)\right|
 \leq |\cos\widehat\theta-\cos\theta|+\frac\eta4
 \leq d_{2\pi}(\widehat\theta,\theta)+\frac\eta4
 \leq\frac\eta2.
\]
On failure the error is at most $2$.  Measuring the reference register in
the $\{\ket{\psi_j}\}$ basis therefore couples the output random variable
$Z\sim\nu$ to $\rho_B$ with expected distance at most
$\frac\eta2+2\frac\eta4=\eta$.
Hence $W_1(\nu,\rho_B)\leq\eta$, and $\nu$ is supported on $2/\eta$
points.

By \Cref{lem:coherent-phase-estimation}, phase estimation and its inverse
use $O(\eta^{-1}\log(1/\eta))$ controlled applications of $Q$ or
$Q^\dagger$.  Each costs $O(\sqrt\Delta)$ graph queries by
\Cref{lem:projected-encoding}, which proves the claimed
$\widetilde O(\eta^{-1}\sqrt\Delta)$ query bound.  Controlling the entire
construction only controls its constituent operations and hence does not
change this bound.
\end{proof}

\subsection{Quantum distribution learning}

By \Cref{lem:purified-spectrum}, we already derive a efficient sampler from the approximate spectral density of $B$. The standard method of learning a probability distribution under $W_1$ requires $O(\eps^{-2})$ such sampling queries~\cite{FournierGuillin2015,Panaretos_2019}. However, in quantum setting, we can improve this to $\widetilde O(\eps^{-1})$ queries. We use the following result of $\ell_\infty$-estimation.

\begin{lemma}\cite[Theorem~9]{vanApeldoorn2021}
\label{lem:multidimensional-ae}
Let $p=(p_1,\ldots,p_d)$ be a probability vector, and suppose a unitary
$O_p$ and its inverse are available, where
\[
                  O_p\ket0=\sum_{i=1}^d
                         \sqrt{p_i}\ket i\ket{\gamma_i}
\]
for arbitrary normalized workspace states $\ket{\gamma_i}$.  For every
$0<\alpha<1$, a classical probability vector $\widetilde p$ satisfying
\[
                         \norm{\widetilde p-p}_\infty\leq\alpha
\]
with success probability $\ge 2/3$ can be found using
$\widetilde O(1/\alpha)$ controlled applications of $O_p$ and $O_p^\dagger$.
\end{lemma}

We change this estimator to the Wasserstein learning from a coherent sampler.

\begin{lemma}
\label{lem:quantum-wasserstein-learning}
Let $U$ be a coherent experiment whose measured output has density function $\nu$ on
$[-1,1]$, and suppose controlled $U$ and $U^\dagger$ are available.  For
every $0<\eps<1$, there is a algorithm using
$\widetilde O(1/\eps)$ calls to $U$ and $U^\dagger$ that outputs a classical
measure $\widehat\nu$, supported on $O(1/\eps)$ points, such that 
\[
                             W_1(\widehat\nu,\nu)\leq\eps.
\]
with success probability $\ge 2/3$.
\end{lemma}

\begin{proof}
Choose $J$ so that $K=2^J\in[4/\eps,8/\eps)$.  Partition $[-1,1]$
dyadically as follows: at level $\ell=1$ split it into intervals $[-1,0]$ and
$[0,1]$; from level $\ell$ to level $\ell+1$, split every current interval
at its midpoint.  After $J$ levels one has $K$ leaf intervals; see
\Cref{fig:dyadic-w1}.  Let $\nu^{(\ell)}_k$ be the mass of $\nu$ of the $k$-th interval at level $\ell$.  On the disjoint union of all intervals we define
\[
                    r_{\ell,k}=\frac{\nu^{(\ell)}_k}{J},
       \qquad 1\leq\ell\leq J,\quad 1\leq k\leq2^\ell.
\]
Then $r:=[r_{\ell,k}]_{1\leq\ell\leq J, 1\leq k\leq2^\ell}$ is a probability vector since the summation of all variables $r_{\ell,k}$ is $1$.

We implement a probability oracle $O_r$ for $r$ as follows.  Start from
clean registers $\ket0_\ell\ket0_Z\ket0_k$ and prepare a uniform
superposition over levels,
\[
 \ket0_\ell\ket0_Z\ket0_k
 \longmapsto
 \frac1{\sqrt J}\sum_{\ell=1}^J
        \ket\ell_\ell\ket0_Z\ket0_k.
\]
Apply the coherent experiment $U$ to the sample register $Z$, obtaining
\[
 \frac1{\sqrt J}\sum_{\ell=1}^J
        \ket\ell_\ell
        \Bigl(\sum_z\sqrt{\nu(z)}\,\ket z_Z\ket{\gamma_z}\Bigr)
        \ket0_k,
\]
where measuring $Z$ would yield a sample from $\nu$ and $\ket{\gamma_z}$
is an arbitrary workspace state.  Write $\operatorname{bin}_\ell(z)$ for
the index of the unique level-$\ell$ interval containing $z$.  After
the affine map $x=(z+1)/2\in[0,1]$, the index of this unique level-$\ell$ interval is exactly 
$\operatorname{bin}_\ell(z)=\lfloor x\cdot 2^\ell\rfloor$.  Copying it into the third register is a
reversible classical circuit on $(\ell,Z)$:
\[
 \frac1{\sqrt J}\sum_{\ell=1}^J\sum_z
        \sqrt{\nu(z)}\,
        \ket\ell_\ell\ket z_Z
        \ket{\operatorname{bin}_\ell(z)}_k
        \ket{\gamma_z}.
\]
Fix a pair $(\ell,k)$.  The corresponding sum is over those $z$ with
$\operatorname{bin}_\ell(z)=k$, and
\[
 \sum_{z:\,\operatorname{bin}_\ell(z)=k}\nu(z)
 =\nu^{(\ell)}_k
 =J\,r_{\ell,k}.
\]
Absorbing $\ket z_Z\ket{\gamma_z}$ into a normalized workspace
$\ket{\gamma_{\ell,k}}$ therefore produces
\[
 \sum_{\ell,k}\sqrt{r_{\ell,k}}
        \ket{\ell,k}\ket{\gamma_{\ell,k}},
\]
which is the form required by \Cref{lem:multidimensional-ae}.  The map
$\operatorname{bin}_\ell$ uses no call to $U$.  Thus one use of $O_r$ or $O_r^\dagger$ costs one use of $U$ or
$U^\dagger$, respectively; the same is true for their controlled versions.

Apply \Cref{lem:multidimensional-ae} with
$\alpha=\eps/(16J^2)$ to obtain $\widetilde r$ satisfying
$\norm{\widetilde r-r}_\infty\leq\alpha$ with probability $\ge 2/3$.  Set
\[
                \widetilde \nu^{(\ell)}_k=J\widetilde r_{\ell,k},
\]
We write \(\widetilde\nu^{(\ell)}=(\widetilde\nu^{(\ell)}_k)_{1\le k\le 2^{\ell}}\) for the estimated \(\ell\)-level mass vector, and write \(\widetilde\nu=(\widetilde\nu^{(J)}_k)_{1\le k\le K}\) for the
estimated leaf-mass vector. We note that they are not necessarily probability vector.
We also write
\(\nu_{\mathrm{grid}}=[\nu^{(J)}_k]_{1\le k\le K}\in\R^K\) for the true leaf-probability vector.  Let \(A_\ell\) aggregate leaf masses into
level-\(\ell\) intervals.  Compute a probability vector
\(\widehat\nu\in\R^K\) minimizing
\[
 \Phi(x)=\sum_{\ell=1}^J2^{-\ell}
               \norm{A_\ell x-\widetilde \nu^{(\ell)}}_1,
 \qquad \text{s.t.}\qquad x\geq0,\quad\sum_i x_i=1.
\]
By $\norm{\widetilde r-r}_\infty\le \alpha$, it holds that 
\[
 \Phi(\nu_{\mathrm{grid}})
 \leq\sum_{\ell=1}^J2^{-\ell}\,2^\ell J\alpha
 =J^2\alpha=\frac\eps{16}.
\]
By optimality of $\widehat\nu$ and the triangle inequality,
\begin{equation}
\label{eq:dyadic-tree-discrepancy}
 \sum_{\ell=1}^J2^{-\ell}
       \norm{A_\ell(\widehat \nu-\nu_{\mathrm{grid}})}_1
 \leq \Phi(\widehat \nu)+\Phi(\nu_{\mathrm{grid}})
 \leq\frac\eps8.
\end{equation}

Now for the $k$-th interval at level $\ell$, we write
\[
 \Delta_{\ell,k}
 :=\widehat\nu^{(\ell)}_k-\nu^{(\ell)}_k,
\]
equivalently the $k$-th entry of $A_\ell(\widehat\nu-\nu_{\mathrm{grid}})$.  Routing the
discrepancy along the tree sends net mass \(\lvert\Delta_{\ell,k}\rvert\)
across the parent edge of that interval, see \Cref{fig:dyadic-w1}.  Define
\[
 d_{\mathrm{tree}}(\widehat \nu,\nu_{\mathrm{grid}})
 :=\sum_{\ell=1}^J 2^{-\ell}
       \norm{A_\ell(\widehat \nu-\nu_{\mathrm{grid}})}_1,
\]
which is the left-hand side of \eqref{eq:dyadic-tree-discrepancy}, hence
$d_{\mathrm{tree}}(\widehat \nu,\nu)\le\eps/8$.  Identify \(\widehat\nu\)
with the discrete measure placing masses \(\widehat\nu_k\) at the leaf
centers.  A level-\(\ell\) interval of \([-1,1]\) has length \(2^{1-\ell}\),
while the corresponding tree edge has length \(2^{-\ell}\), so the
Euclidean distance between two leaf centers is at most twice their tree
distance.  Thus this tree routing is a coupling of \(\widehat\nu\) and \(\nu_{\mathrm{grid}}\), and
\[
 W_1(\widehat\nu,\nu_{\mathrm{grid}})
 \le 2\,d_{\mathrm{tree}}(\widehat\nu,\nu_{\mathrm{grid}})
 \le\frac\eps4.
\]
Moving every sample of \(\nu\) to the center of its leaf interval costs at
most \(1/K\le\eps/4\).  Therefore
\[
 W_1(\widehat\nu,\nu)
 \le W_1(\widehat\nu,\nu_{\mathrm{grid}})
      +W_1(\nu_{\mathrm{grid}},\nu)
 \le\frac\eps4+\frac\eps4
 =\frac\eps2<\eps.
\]

The index set of $r$ has $\sum_{\ell=1}^J2^\ell=O(K)=O(1/\eps)$ elements,
and \Cref{lem:multidimensional-ae} uses
\[
          \widetilde O(1/\alpha)
          =\widetilde O(J^2/\eps)
          =\widetilde O(1/\eps)
\]
oracle calls.  Finally, computing $\widehat\nu$ is a linear program of size
$O(1/\eps)$, which doesn't need any graph queries and can be solved in $\widetilde{O}(\eps^{-2})$ additional time using minimum-cost flow on trees \cite{orlin1988faster}.
\end{proof}

\begin{figure}[t]
\centering
\begin{tikzpicture}[font=\small, >=Stealth]
  \def\W{11.2}
  \def\n{8}
  \pgfmathsetmacro{\dx}{\W/\n}
  \foreach \i in {0,...,7} {
    \pgfmathsetmacro{\cx}{(\i+0.5)*\dx}
    \coordinate (L\i) at (\cx,0);
  }
  \def\hiii{1.15}
  \def\hii{2.15}
  \def\hi{3.15}
  \def\hroot{3.85}

  \fill[blue!10] (4*\dx, -0.22) rectangle (\W, \hi+0.08);
  \fill[blue!16] (4*\dx, -0.22) rectangle (6*\dx, \hii+0.08);
  \fill[blue!25] (4*\dx, -0.22) rectangle (5*\dx, \hiii+0.08);

  \foreach \i/\j in {0/1, 2/3, 4/5, 6/7} {
    \pgfmathsetmacro{\xm}{(\i+\j+1)*\dx/2}
    \draw[thick] (L\i) -- (\i*\dx+\dx/2, \hiii) -- (\j*\dx+\dx/2, \hiii) -- (L\j);
    \fill[black] (\xm, \hiii) circle (1.4pt);
  }
  \foreach \a/\b/\c/\d in {0/1/2/3, 4/5/6/7} {
    \pgfmathsetmacro{\xab}{(\a+\b+1)*\dx/2}
    \pgfmathsetmacro{\xcd}{(\c+\d+1)*\dx/2}
    \pgfmathsetmacro{\xm}{(\a+\d+1)*\dx/2}
    \draw[thick] (\xab, \hiii) -- (\xab, \hii) -- (\xcd, \hii) -- (\xcd, \hiii);
    \fill[black] (\xm, \hii) circle (1.4pt);
  }
  \pgfmathsetmacro{\xl}{(0+3+1)*\dx/2}
  \pgfmathsetmacro{\xr}{(4+7+1)*\dx/2}
  \draw[thick] (\xl, \hii) -- (\xl, \hi) -- (\xr, \hi) -- (\xr, \hii);
  \fill[black] ({\W/2}, \hi) circle (1.4pt);
  \draw[thick] ({\W/2}, \hi) -- ({\W/2}, \hroot);
  \fill[black] ({\W/2}, \hroot) circle (1.6pt);

  \draw[line width=1.15pt, blue!70!black]
    (L4) -- (4.5*\dx, \hiii) -- (5*\dx, \hiii) -- (5*\dx, \hii)
    -- (6*\dx, \hii) -- (6*\dx, \hi) -- ({\W/2}, \hi) -- ({\W/2}, \hroot);

  \draw[thick] (0,0) -- (\W,0);
  \foreach \i in {0,...,8} {
    \draw[thick] ({\i*\dx}, -0.12) -- ({\i*\dx}, 0.12);
  }
  \foreach \i in {0,...,7} {
    \fill (L\i) circle (1.5pt);
  }
  \node[anchor=east] at (0,-0.42) {$-1$};
  \node[anchor=west] at (\W,-0.42) {$1$};

  \draw[blue!70!black, dashed, thick] (4*\dx, -0.22) -- (4*\dx, \hiii+0.35);
  \draw[<->, blue!70!black] (4*\dx, \hiii+0.28) -- (5*\dx, \hiii+0.28);
  \node[blue!70!black, anchor=south, inner sep=1pt] at (4.5*\dx, \hiii+0.32)
    {$2^{-3}\lvert\Delta_{3,5}\rvert$};

  \draw[blue!60!black, dashed] (4*\dx, \hii+0.12) -- (4*\dx, \hii+0.55);
  \draw[blue!60!black, dashed] (6*\dx, \hii+0.12) -- (6*\dx, \hii+0.55);
  \draw[<->, blue!60!black] (4*\dx, \hii+0.48) -- (6*\dx, \hii+0.48);
  \node[blue!60!black, anchor=south, inner sep=1pt] at (5*\dx, \hii+0.52)
    {$2^{-2}\lvert\Delta_{2,3}\rvert$};

  \draw[<->] (\W/2+0.15, \hi+0.42) -- (\W, \hi+0.42);
  \node[anchor=south, inner sep=1pt] at (0.75*\W, \hi+0.46)
    {$2^{-1}\lvert\Delta_{1,2}\rvert$};

  \node[anchor=east] at (-0.15, \hi) {$\ell=1$};
  \node[anchor=east] at (-0.15, \hii) {$\ell=2$};
  \node[anchor=east] at (-0.15, \hiii) {$\ell=3$};
  \node[anchor=east] at (-0.15, 0) {leaves};

  \node at ({\W/2}, \hroot+0.38)
    {$d_{\mathrm{tree}}(\widehat \nu,\nu_{\mathrm{grid}})=\sum_{\ell}2^{-\ell}\lVert A_\ell(\widehat \nu-\nu_{\mathrm{grid}})\rVert_1$};
\end{tikzpicture}
\caption{Dyadic intervals of $[-1,1]$.  The $k$-th interval at level $\ell$ is
joined to its parent by an edge of length $2^{-\ell}$.  Its imbalance
$\Delta_{\ell,k}$ is the difference between the masses that
$\widehat\nu$ and $\nu$ assign to that interval, equivalently an entry of
$A_\ell(\widehat\nu-\nu_{\mathrm{grid}})$.  The net flow across the parent edge equals
$\lvert\Delta_{\ell,k}\rvert$ and contributes
$2^{-\ell}\lvert\Delta_{\ell,k}\rvert$ to
$d_{\mathrm{tree}}(\widehat\nu,\nu_{\mathrm{grid}})$.  The resulting tree routing is a
coupling of $\widehat\nu$ and $\nu_{\mathrm{grid}}$, and
$W_1(\widehat\nu,\nu_{\mathrm{grid}})\le 2d_{\mathrm{tree}}(\widehat\nu,\nu_{\mathrm{grid}})$.}
\label{fig:dyadic-w1}
\end{figure}

\begin{proof}[Proof of \Cref{thm:intro-quantum-upper}.]
It suffices to consider $0<\eps<1$.  Set
$\Delta=\lceil32/\eps^2\rceil$.  By \Cref{lem:truncation},
$W_1(\rho_N,\rho_B)\leq\eps/4$.  Run
\Cref{lem:purified-spectrum} with $\eta=\eps/4$, so its output law $\nu$
satisfies $W_1(\nu,\rho_B)\leq\eps/4$.  Finally apply
\Cref{lem:quantum-wasserstein-learning} to $U_{\rm spec}$ with target
accuracy $\eps/2$.  Since $\rho_N=\rho_G$, the triangle inequality gives
\[
                         W_1(\widehat\rho,\rho_G)\leq\eps,
\]
where $\widehat\rho$ is the classical measure returned by
\Cref{lem:quantum-wasserstein-learning}.  Amplifying the constant-success
guarantee of \Cref{lem:multidimensional-ae} makes the overall success
probability at least $2/3$.

The learning algorithm makes $\widetilde O(1/\eps)$ calls to
$U_{\rm spec}$ and $U_{\rm spec}^\dagger$.  By
\Cref{lem:purified-spectrum}, each such call costs
\[
       \widetilde O(\eta^{-1}\sqrt\Delta)
       =\widetilde O(\eps^{-2})
\]
graph queries.  The total graph query complexity is therefore
$\widetilde O(\eps^{-3})$.  The output has $O(1/\eps)$ support points, and
the classical reconstruction via linear program talkes $\widetilde{O}(\eps^{-2})$ time, see the proof of \Cref{lem:quantum-wasserstein-learning}.
\end{proof}

\section{A quantum query lower bound}
\label{sec:quantum-lower}

Next we prove \Cref{thm:intro-quantum-lower}. We use a family of parity-controlled cycle graphs.  Its spectral densities form a large packing in $W_1$, so estimating the spectral density of one graph in this family identifies an entire codeword. Finally, we use the polynomial method to prove the desired lower bound.

We begin with some properties of cycle graph. Denote $C_m$ by the length $m$ cycle graph. For an integer $m\geq5$, let $\mu_m=\rho_{C_m}$.  The eigenvalues of the normalized adjacency matrix of $C_m$ are $\cos(2\pi k/m)$, $0\leq k<m$ \cite[Theorem 5.5.1]{spielman2019sagt}.  Define
\[
                    q_m(\theta)=\operatorname{sgn}(\sin(m\theta)),
                    \qquad 0<\theta<\pi,
\]
where $\operatorname{sgn}(x)$ is the sign function, $\operatorname{sgn}(x)=1$ if $x>0$, $\operatorname{sgn}(x)=0$ if $x=0$, and $\operatorname{sgn}(x)=-1$ if $x<0$. We give the following Lemma about the spectral difference of $C_m$ and $C_{2m}$.

\begin{lemma}
\label{lem:cycle-spectral-difference}
For almost every $\theta\in(0,\pi)$,
\[
 F_{\mu_{2m}}(\cos\theta)-F_{\mu_m}(\cos\theta)
                         =\frac{q_m(\theta)}{2m}.
\]
\end{lemma}

\begin{proof}
The eigenvalues of $C_r$ are $\cos(2\pi k/r)$ for $0\le k<r$.  For
$\theta\in(0,\pi)$ with $\cos\theta$ not an eigenvalue,
$\cos(2\pi k/r)\ge\cos\theta$ if and only if
$k\le r\theta/(2\pi)$ or $k\ge r-r\theta/(2\pi)$.  Hence
\[
 \mu_r([\cos\theta,1])
       =\frac{1+2\lfloor r\theta/(2\pi)\rfloor}{r}.
\]
 Set $y=m\theta/\pi$, then it holds that
\begin{align*}
 F_{\mu_{2m}}(\cos\theta)-F_{\mu_m}(\cos\theta)
 &=\mu_m([\cos\theta,1])-\mu_{2m}([\cos\theta,1])\\
 &=\frac{1+2\lfloor y/2\rfloor}{m}
   -\frac{1+2\lfloor y\rfloor}{2m}
 =\frac{1+4\lfloor y/2\rfloor-2\lfloor y\rfloor}{2m}.
\end{align*}
If $\lfloor y\rfloor=2k$ is even then
$\lfloor y/2\rfloor=k$ and the numerator is $1$; if $\lfloor y\rfloor=2k+1$ is odd
then $\lfloor y/2\rfloor=k$ and the numerator is $-1$.  For
non-integral $y$, $\operatorname{sgn}(\sin(m\theta))=\operatorname{sgn}(\sin(\pi y))=(-1)^{\lfloor y\rfloor}$.
Thus the difference equals $q_m(\theta)/(2m)$.
\end{proof}

\subsection{The construction of hard instance.}
Next we discribe our construction. We consider the vertex set $\mathbb Z_m\times\{0,1\}$. Fix a large integer $L$ and let
\[
 \mathcal M=[L,2L]\cap(2\mathbb Z+1),\qquad K=|\mathcal M|=\Theta(L).
\]
For each $m\in\mathcal M$, take an input block
$x^{(m)}\in\{0,1\}^{L}$ and extend it by zeros to
$y^{(m)}\in\{0,1\}^{m}$.  On $\mathbb Z_m\times\{0,1\}$ put the edges
\[
 (t,s)\sim(t+1,s\oplus y^{(m)}_t).
\]

We call the resulting graph $G_{x^{(m)}}$. The vertex set $\mathbb Z_m\times\{0,1\}$ consists of $m$ columns and two
sheets, hence $2m$ vertices of degree two, so $G_{x^{(m)}}$ is a disjoint union
of cycles.  Starting at $(0,0)$ and walking once around the $m$ columns
ends on sheet $\bigoplus_{t=0}^{m-1}y^{(m)}_t$.  Since $y^{(m)}$ is
$x^{(m)}$ padded by zeros, this total twist equals
$b_m:=\bigoplus_{t=0}^{L-1}x^{(m)}_t$.  If $b_m=0$, one lap returns to the
same sheet and the two sheets form separate $m$-cycles; if $b_m=1$, one
lap switches sheets and two laps are needed to close, giving a single
$2m$-cycle.
As a result, the constructed graph $G_{x^{(m)}}=C_m\sqcup C_m$ if $b_m=0$ and $G_{x^{(m)}}=C_{2m}$ if $b_m=1$.  Let $G_x=\sqcup_{m\in\mathcal M}G_{x^{(m)}}$ and write
$S=\sum_{m\in\mathcal M}m$.  Thus $G_x$ has $2S=\Theta(KL)$ vertices, and $G_x$ is uniquely determined by $b\in \{0,1\}^{\mathcal M}$. We write $\rho_b=\rho_{G_x}$.

\subsection{Spectral density seperation between $G_x$.}
Next we exhibit a large packing $\mathcal Z\subseteq\{0,1\}^{\mathcal M}$
such that the spectral densities of the associated graphs $G_z$ are
pairwise separated in Wasserstein--$1$ distance.  This is the content of
the following lemma.

\begin{lemma}
  \label{lem:square-wave-packing}
  There are absolute constants $c,C>0$ and a set
  $\mathcal Z\subseteq\{0,1\}^{\mathcal M}$ satisfying
  \[
   |\mathcal Z|\geq\exp\!\left(\frac{cK}{B_L}\right),
   \qquad B_L=C\log\log(3L),
  \]
  such that, for all distinct $z,z'\in\mathcal Z$,
  \[
   W_1(\rho_z,\rho_{z'}) \geq \frac{c\sqrt K}{S}=\Omega(L^{-3/2}).
  \]
  \end{lemma}

Some preliminaries are needed for the proof of \Cref{lem:square-wave-packing}.

\begin{lemma}
\label{lem:spectral-packing-formula}
For parity vectors $b,c\in\{0,1\}^{\mathcal M}$,
\[
 W_1(\rho_b,\rho_c)
 =\frac1{2S}\int_0^\pi
       \left|\sum_{m\in\mathcal M}(b_m-c_m)q_m(\theta)\right|
       \sin\theta\,d\theta .
\]
\end{lemma}

\begin{proof}
$G_{x^{(m)}}$ has $2m$ vertices, hence weight $m/S$ in the spectral density.
The CDF identity in \Cref{lem:cycle-spectral-difference} gives
\[
 F_{\rho_b}(\cos\theta)-F_{\rho_c}(\cos\theta)
   =\frac1{2S}\sum_m(b_m-c_m)q_m(\theta).
\]
Now using \Cref{prop:kr-duality} and the substitution $x=\cos\theta$ gives the desired result.
\end{proof}

The square-wave functions $q_m$ have an explicit Gram identity, which
follows from the classical Fourier series of $\operatorname{sgn}(\sin x)$
\cite[Ch.~I]{Zygmund2002}.

\begin{lemma}
\label{lem:square-wave-gram}
For odd positive integers $m,m'$,
\[
 \frac1\pi\int_0^\pi q_m(\theta)q_{m'}(\theta)\,d\theta
             =\frac{\gcd(m,m')^2}{mm'}.
\label{eq:square-wave-gram}
\]
\end{lemma}

\begin{proof}
The square-wave expansion \cite[Ch.~I]{Zygmund2002} gives
\[
  q_j(\theta)=\operatorname{sgn}(\sin j\theta)
 =\frac4\pi
 \sum_{\substack{r\geq1\\ r\text{ odd}}}\frac{\sin(rj\theta)}{r}.
\]

Let $d=\gcd(m,m')$, $m=da$, and $m'=db$.  Orthogonality of sines on
$(0,\pi)$ retains only coincident frequencies, indexed by $r=b\ell$ and
$s=a\ell$ with $\ell$ odd.  Since
$\sum_{\ell\text{ odd}}\ell^{-2}=\pi^2/8$, the inner product is
$1/(ab)=d^2/(mm')$.
\end{proof}

We also utilize the following inequalities.

\begin{lemma}[Paley--Zygmund inequality~\cite{PaleyZygmund1932}]
\label{lem:paley-zygmund}
Let $Z\ge0$ be a random variable with $\mathbb E Z^2<\infty$.  For every
$\theta\in[0,1]$,
\[
 \Pr(Z>\theta\,\mathbb E Z)
 \ge(1-\theta)^2\frac{(\mathbb E Z)^2}{\mathbb E Z^2}.
\]
\end{lemma}

\begin{lemma}[Talagrand's convex Lipschitz concentration~\cite{Talagrand1995}]
\label{lem:talagrand}
Let $f:\mathbb R^n\to\mathbb R$ be convex and $\lambda$-Lipschitz with
respect to the Euclidean norm.  If $X$ is uniformly random in
$\{0,1\}^n$, then
\[
 \Pr\bigl(|f(X)-\mathbb E f|\ge t\bigr)
 \le 2\exp(-ct^2/\lambda^2)
\]
for an absolute constant $c>0$.
\end{lemma}

\begin{proof}[Proof of \Cref{lem:square-wave-packing}.]
Let $H$ be the Gram matrix. That is, $H_{m,m'}=\frac1\pi\int_0^\pi q_m(\theta)q_{m'}(\theta)\,d\theta$ for each entries of $H$. By \Cref{lem:square-wave-gram}, we have $H_{m,m'}=\frac{\gcd(m,m')^2}{mm'}$. We bound the $\ell_2$ norm of $H$. To this end, we compute the $\ell_\infty$ norm. Every $m'\in\mathcal M$ satisfies $m'\ge L$, hence
\[
 \sum_{m'\in\mathcal M}H_{m,m'}
 =\sum_{m'\in\mathcal M}\frac{\gcd(m,m')^2}{mm'}
 \le\frac1{mL}\sum_{m'\in\mathcal M}\gcd(m,m')^2.
\]
Grouping the inner sum by $d=\gcd(m,m')$ and using that there are at most
$L/d+1$ multiples of $d$ in $[L,2L]$ gives
\[
 \sum_{m'\in\mathcal M}\gcd(m,m')^2
 =\sum_{d\mid m}d^2\bigl|\{m'\in\mathcal M:\gcd(m,m')=d\}\bigr|
 \le\sum_{d\mid m}d^2\Bigl(\frac Ld+1\Bigr).
\]
Therefore
\begin{align*}
 \sum_{m'\in\mathcal M}H_{m,m'}
 &\leq \frac1{mL}\sum_{d\mid m}d^2\left(\frac Ld+1\right)
 =\frac1m\sum_{d\mid m}d+\frac1{mL}\sum_{d\mid m}d^2
 =O(\log\log(3L)).
\end{align*}
Here we used the standard bound
$\frac1m\sum_{d\mid m}d=O(\log\log(3m))$ \cite[Thm.~323]{HardyWright2008}.  Since $H$ is symmetric and nonnegative,
$\|H\|_2\le \|H\|_\infty\le B_L$ after increasing $C$.

For independent uniform $z,z'\in\{0,1\}^{\mathcal M}$, set
$w_m=z_m-z'_m$.  For almost every fixed $\theta$, the variables
$w_mq_m(\theta)$ are independent and have probabilities $1/2,1/4,1/4$ at
$0,1,-1$.  Write $Y_m=w_mq_m(\theta)$ and $X:=\sum_m w_mq_m(\theta)$.  Then the $Y_m$
are i.i.d.\ with $\mathbb E Y_m=0$ and $\mathbb E Y_m^2=\mathbb E Y_m^4=1/2$.
Independence therefore gives $\mathbb E X^2=K/2$ and
\[
 \mathbb E X^4
 =K\,\mathbb E Y_1^4+3K(K-1)(\mathbb E Y_1^2)^2
 =\frac K2+\frac{3K(K-1)}4
 \leq\frac{3K^2}4.
\]
\Cref{lem:paley-zygmund} applied to $X^2$ therefore gives $\Prb(X^2\ge \frac{1}{2} \mathbb E(X^2))\ge \frac{(\mathbb EX^2)^2}{4\mathbb E(X^4)}\ge \frac{1}{12}$. So
\[
  \mathbb E\left|\sum_m w_mq_m(\theta)\right|=\mathbb{E}\left|X\right|
  \ge \sqrt{\frac{K}{4}}\cdot \Prb\left(X^2\ge \frac{K}{4}\right)\ge  c\sqrt K.
\]
This lower bound is independent of $\theta$.  By
\Cref{lem:spectral-packing-formula} and Fubini's Theorem,
\begin{align*}
 \mathbb E\bigl[W_1(\rho_z,\rho_{z'})\bigr]
 &=\frac1{2S}\int_0^\pi
      \mathbb E\Bigl|\sum_m w_mq_m(\theta)\Bigr|
      \sin\theta\,d\theta
 \geq \frac{c\sqrt K}{2S}\int_0^\pi\sin\theta\,d\theta
 =\frac{c\sqrt K}{S}.
\end{align*}

We write $\Phi(d)=\int_0^\pi|\sum_m d_mq_m(\theta)|\sin\theta\,d\theta$. By \Cref{lem:spectral-packing-formula},
$W_1(\rho_z,\rho_{z'})=\frac1{2S}\Phi(z-z')$.
The map $d\mapsto\Phi(d)$ is convex, hence $(z,z')\mapsto W_1(\rho_z,\rho_{z'})$
is convex on $\mathbb R^{2K}$.  For coefficient vectors $d,e$, Cauchy--Schwarz inequality gives
\[
 \left|\Phi(d)-\Phi(e)\right|
 \leq \pi\sqrt{(d-e)^{\mathsf T}H(d-e)}
 \leq \pi\sqrt{B_L}\,\|d-e\|_2,
\]
so $\Phi$ is $O(\sqrt{B_L})$-Lipschitz.  Since $d=z-z'$, the function
$(z,z')\mapsto W_1(\rho_z,\rho_{z'})$ is therefore
$O(\sqrt{B_L}/S)$-Lipschitz. We use \Cref{lem:talagrand}. Write $f(z,z')=W_1(\rho_z,\rho_{z'})$ on the
cube $\{0,1\}^{2K}$.  Then $\mathbb E f\ge c\sqrt K/S$ and
$\lambda=O(\sqrt{B_L}/S)$, so a deviation $t=\Theta(\sqrt K/S)$ satisfies
$t^2/\lambda^2=\Theta(K/B_L)$.  \Cref{lem:talagrand} therefore gives
\[
 \Pr\Bigl\{W_1(\rho_z,\rho_{z'})<\frac{c\sqrt K}S\Bigr\}
 \leq \exp(-cK/B_L).
\]
Choosing $\exp(cK/(3B_L))$ random vectors and taking a union bound over all
pairs proves the lemma.
\end{proof}

\subsection{Query lower bound via polynomial method.}

From the query lower bound aspect, we first prove the query bound includes the cost of computing the block parities. Then we prove the query bound includes the cost of graph queries. We first formally define the following quantum bit oracle.

\begin{definition}
  \label{def:block-bit-oracle}
  Access to the blocks is the XOR bit oracle on the concatenated string
  $x=(x^{(m)})_{m\in\mathcal M}\in\{0,1\}^{KL}$:
  \[
   O_x\ket{m,t,b}
   =\ket{m,t,\,b\oplus x_t^{(m)}},
   \qquad t\in\{0,\ldots,L-1\}.
  \]
  The registers $m,t,b$ may be in superposition.  
  \end{definition}

\begin{lemma}
\label{lem:codeword-identification}
Given bit-oracle access to blocks $x^{(m)}\in\{0,1\}^L$, under the promise that
\[
 z_m=\bigoplus_{t=0}^{L-1}x_t^{(m)}
 \quad\text{and}\quad z=(z_m)_{m\in\mathcal M}\in\mathcal Z, \quad |\mathcal Z|\geq\exp\!\left(\frac{cK}{B_L}\right)
\]
outputting $z$ with probability at least $2/3$ requires
\[
             \Omega\!\left(\frac{LK}{B_L\log(2+B_L)}\right)
\]
quantum queries.
\end{lemma}

We use the following polynomial method to prove \Cref{lem:codeword-identification}.

\begin{lemma}\cite{BealsEtAl2001}
  \label{lem:polynomial-method}
  Let $x\in\{0,1\}^N$ and write $\xi_i=(-1)^{x_i}$.  If a quantum algorithm
  makes $T$ queries to the XOR bit oracle for $x$, then for every output
  $a$ the probability $p_a(\xi)$ of outputting $a$ is a multilinear
  polynomial in $\xi$ of degree at most $2T$.
  \end{lemma}

\begin{proof}[Proof of \Cref{lem:codeword-identification}.]
Write $\xi_t^{(m)}=(-1)^{x_t^{(m)}}$.  If an algorithm makes $T$ queries, then
by \Cref{lem:polynomial-method}, for every $a\in\mathcal Z$ its output
probability $p_a(\xi)$ is a multilinear polynomial of degree at most $2T$.
For $z\in\mathcal Z$,
average this polynomial uniformly over the corresponding parity fiber:
\[
 \overline p_a(z):=\mathbb E\!\left[p_a(\xi)\,\middle|\,
        \prod_{t=0}^{L-1}\xi_t^{(m)}=(-1)^{z_m}\ \text{for every }m\right].
\]
For every $S\subseteq\{0,\ldots,L-1\}$,
\[
 \mathbb E\!\left[\prod_{t\in S}\xi_t^{(m)}\,\middle|\,
             \prod_{t=0}^{L-1}\xi_t^{(m)}=(-1)^{z_m}\right]
 =\begin{cases}
    1,&S=\varnothing,\\
    (-1)^{z_m},&S=\{0,\ldots,L-1\},\\
    0,&\text{otherwise}.
  \end{cases}
\]
Hence a monomial survives the averaging only if it uses either none or all
$L$ variables of each block.  Therefore $\overline p_a$ has degree at most
$d=\lfloor2T/L\rfloor$ in $z$.

Let $P$ be the matrix indexed by $z,a\in\mathcal Z$ with
$P_{z,a}=\overline p_a(z)$.  Pointwise success on every promised input gives
$P_{z,z}\geq2/3$ and $\sum_{a\neq z}P_{z,a}\leq1/3$.  Thus $P$ is strictly
row-diagonally dominant and has full rank $|\mathcal Z|$.  On the other hand,
every column of $P$ lies in the evaluation space on $\mathcal Z$ of Boolean
polynomials of degree at most $d$.  Consequently
\[
 |\mathcal Z|=\operatorname{rank}P
 \leq D_d:=\sum_{j=0}^d\binom Kj.
\]
If $d\leq K/2$, then $D_d\leq(eK/d)^d$. By $|\mathcal Z|\geq\exp\!\left(\frac{cK}{B_L}\right)$, this gives
$d\log(eK/d)\geq cK/B_L$. Therefore
\[
 d=\Omega\!\left(\frac{K}{B_L\log(2+B_L)}\right).
\]
Since $T\geq Ld/2$, this proves \Cref{lem:codeword-identification} for $d\leq K/2$. If $d>K/2$ the desired bound is immediate.
\end{proof}

Next we show that quantum graph queries can be simulated using quantum bit queries.

\begin{lemma}
\label{lem:cycle-lift-simulation}
Every quantum degree, indexed-neighbor, or vertex-pair query to $G_x$ can be
simulated exactly using $O(1)$ queries to the bits of the blocks $x^{(m)}$.
\end{lemma}

\begin{proof}
Encode a vertex of $G_x$ as a triple $(m,t,s)$ with $m\in\mathcal M$,
$t\in\mathbb Z_m$, and $s\in\{0,1\}$.  Write
$y^{(m)}_t=x_t^{(m)}$ if $0\le t<L$ and $y^{(m)}_t=0$ if $L\le t<m$.
To load $y^{(m)}_t$ into a clean bit, test whether $t<L$ reversibly: if so,
one call to $O_x$ from \Cref{def:block-bit-oracle} writes $x_t^{(m)}$;
otherwise write $0$. All steps below are reversible and therefore exact on
superpositions. We generate the graph queries by the following:

\begin{itemize}
\item \textbf{Degree query.}  Every vertex has degree two, so
$O_{\rm deg}\ket{v,z}=\ket{v,z\oplus 2}$ uses no call to $O_x$.

\item \textbf{Indexed neighbor query.}  Order the two ports as forward and backward along the cycle:
\[
 \operatorname{nbr}((m,t,s),0)=(m,t+1,s\oplus y^{(m)}_t),
 \qquad
 \operatorname{nbr}((m,t,s),1)=(m,t-1,s\oplus y^{(m)}_{t-1}).
\]
Load the corresponding twist bit by the procedure above, XOR the resulting
vertex label into the answer register, and uncompute the twist bit.  This
implements $O_{\rm nbr}$ with $O(1)$ calls to $O_x$.

\item \textbf{Vertex pair query.}  For $u=(m,t,s)$ and $v=(m',t',s')$. $A_{uv}=1$ if and only if $m=m'$ and either
$(t',s')=(t+1,s\oplus y^{(m)}_t)$ or $(t',s')=(t-1,s\oplus y^{(m)}_{t-1})$. If $m\neq m'$ or the
columns are not consecutive, output $0$ with no query.  Otherwise load the
single relevant twist bit, compare the sheets, XOR the bit $A_{uv}$ into the
answer register, and uncompute.  This implements $O_{\rm pair}$ with $O(1)$
calls to $O_x$.
\end{itemize}
\end{proof}

\begin{proof}[Proof of \Cref{thm:intro-quantum-lower}.]
Let $z_m=b_m$ be the parity of block $x^{(m)}$.  By
\Cref{lem:square-wave-packing}, distinct
$z,z'\in\mathcal Z$ give
\[
                    W_1(\rho_z,\rho_{z'})\geq \frac{c}{L\sqrt K}
                                                =\Omega(L^{-3/2}).
\]
Thus an estimator with error at most a sufficiently small multiple of
$L^{-3/2}$ and success probability at least $2/3$ identifies $z$ by
nearest-neighbor decoding with probability at least $2/3$.
Nearest-neighbor decoding followed by
\Cref{lem:codeword-identification} gives the bit query lower bound
\[
 \Omega\!\left(\frac{LK}{B_L\log(2+B_L)}\right)
 =\Omega\!\left(
      \frac{L^2}{\log\log L\,\log\log\log L}
   \right).
\]
By \Cref{lem:cycle-lift-simulation}, the same bound, up to a constant,
holds for graph queries. Finally, we choose $L^{-3/2}=\Theta(\eps)$ therefore $L=\Theta(\eps^{-2/3})$. So the query lower bound is $\widetilde\Omega(\eps^{-4/3})$. This proves \Cref{thm:intro-quantum-lower}.
\end{proof}



\section{Discussion and open directions}
\label{sec:discussion}

We have characterized the classical local-query landscape: together with
the algorithm of~\cite{CohenSteinerEtAl2018}, the local-query complexity of spectral density estimation is $2^{\Theta(1/\eps)}$.  On the quantum side we proved
that the same problem admits a $\mathrm{poly}(1/\eps)$-query algorithm, an
exponential improvement in $1/\eps$. Matching the quantum upper and lower bounds remains open: our estimator
uses $\widetilde O(\eps^{-3})$ queries, while the lower bound is
$\widetilde\Omega(\eps^{-4/3})$. Another
direction is to design more practical classic or quantum algorithms for
spectral density estimation.

\bibliographystyle{alpha}
\bibliography{ref}
\end{document}